\documentclass[journal,12pt,draftclsnofoot,onecolumn,twoside]{IEEEtran}
\usepackage{times,amsmath,epsfig,amssymb,cite,bm,float,epstopdf,graphicx,graphics,amsthm}
\usepackage{dblfloatfix}
\usepackage{balance}
\usepackage[stable]{footmisc}
\usepackage{tikz,pgfplots}
\usetikzlibrary{fpu}
\usepackage[stable]{footmisc}
\usepackage{lettrine}
\usepackage[font={footnotesize},labelfont={footnotesize}]{caption}
\usepackage{algorithm,algpseudocode}
\usepackage{tabularx,multirow,booktabs}
\usepackage{rotating,afterpage,lscape}
\usepackage{subfigure,float}
\usepackage{adjustbox}
\usepackage{enumitem}
\usepackage{etoolbox}
\AfterEndEnvironment{figure}{\vskip-1ex}

\newtheorem{theorem}{Theorem}
\newtheorem{corollary}{Corollary}[theorem]

\tikzset{>=latex}
\newcolumntype{C}{>{\centering\arraybackslash}X} 
\title{Performance Analysis of Dual-Hop Underwater Wireless Optical Communication
Systems over Mixture Exponential-Generalized Gamma Turbulence Channels}

\small{\author{Emna~Zedini,~\IEEEmembership{Member,~IEEE,}~Abla~Kammoun,~\IEEEmembership{Member,~IEEE,}\\Hamza~Soury,~\IEEEmembership{Member,~IEEE,}~Mounir~Hamdi,~\IEEEmembership{Fellow,~IEEE}~and ~Mohamed-Slim~Alouini,~\IEEEmembership{Fellow,~IEEE}
\thanks{E. Zedini, A. Kammoun and M.-S. Alouini are with the Computer, Electrical, and Mathematical Science and Engineering (CEMSE) Division, King Abdullah University of Science and Technology (KAUST) Thuwal, Makkah Province, Saudi Arabia (e-mails:\{emna.zedini,abla.kammoun,slim.alouini\}@kaust.edu.sa).}
\thanks{H. Soury with the Department of Electrical and Computer Engineering
University of Illinois at Chicago, Chicago, IL, USA (e-mail:\{soury\}@uic.edu).}
\thanks{M. Hamdi is with the College of Science and Engineering, Hamad Bin
Khalifa University (HBKU), Doha, Qatar (e-mail:\{mhamdi\}@hbku.edu.qa).}
}}

\begin{document}
\pgfkeys{/pgf/fpu}
\pgfmathparse{16383+1}
\edef\tmp{\pgfmathresult}
\pgfkeys{/pgf/fpu=false}

\newcommand{\boundellipse}[3]
{(#1) ellipse (#2 and #3)
}
\maketitle
\begin{abstract}
In this work, we present a unified framework for the performance analysis of dual-hop underwater wireless optical communication (UWOC) systems with amplify-and-forward fixed gain relays in the presence of air bubbles and temperature gradients. Operating under either heterodyne detection or intensity modulation with direct detection, the UWOC is modeled by the unified mixture Exponential-Generalized Gamma distribution that we have proposed based on an experiment conducted in an indoor laboratory setup and has been shown to
provide an excellent fit with the measured data under the considered lab channel scenarios.
More specifically, we derive the cumulative distribution function (CDF) and the probability density function of the end-to-end signal-to-noise ratio (SNR) in exact closed-form in terms of the bivariate Fox's H function. Based on this CDF expression, we present novel results for the fundamental performance metrics such as the outage probability, the average bit-error rate (BER) for various modulation schemes, and the ergodic capacity. Additionally, very tight asymptotic results for the outage probability and the average BER at high SNR are obtained in terms of simple functions. Furthermore, we demonstrate that the dual-hop UWOC system can effectively mitigate the short range and both temperature gradients and air bubbles induced turbulences, as compared to the single UWOC link.
All the results are verified via computer-based Monte-Carlo simulations.
\end{abstract}

\begin{IEEEkeywords}
Underwater wireless optical communication (UWOC), air bubbles, temperature gradient, dual-hop relaying, mixture models, Exponential-Generalized Gamma (EGG) distribution, outage probability, bit-error rate (BER), ergodic capacity.
\end{IEEEkeywords}

\section{Introduction}
Underwater wireless optical communication (UWOC) has gained a significant research attention as an appropriate and efficient transmission solution for a variety of underwater applications including offshore oil field exploration, oceanographic data collection, maritime archaeology, environmental monitoring, disaster prevention, and port security among others \cite{SurveyZeng}.
This rapidly growing interest stems from the recent advances in signal processing, digital communication, and low-cost visible light-emitting diodes (LEDs) and laser diodes (LD) that have the lowest attenuation in seawater \cite{LEDadvances, HassanHighDR, OubeiQAMOFDM, Oubei20m}. UWOC systems, operating in the blue/green portion of the spectrum in the 400-550 nm wavelength band, promise high data rates, low-latency, high transmission security, and reduced energy
consumption, compared with their acoustic counterparts \cite{SurveyZeng, Simodetection, HighBdW}. However, UWOC systems suffer from severe absorption and scattering introduced by the underwater channel \cite{SurveyZeng,Jaruwatanadilok,Gabriel2,KihongWCL,Abla19,Mahdi19} as well as underwater optical turbulence (UOT) that results from rapid changes in the refractive index of the water caused by temperature fluctuations, salinity variations as well as the presence of air bubbles in seawater \cite{UOT,Nikishov2000,Korotkova2012,Hagem,OubeiBubbles}. As a consequence, the received optical intensity
undergoes rapid fluctuations which may degrade the UWOC
system performance and affect its reliability.

Investigating the proper statistical distribution of optical signal fluctuations due to UOT is a fundamental challenge in UWOC.
Early studies on UOT have used free-space atmospheric turbulence models such as the Lognormal distribution to describe the irradiance fluctuations in the underwater environment \cite{Gercekcioglu,lognmunderwater}. In \cite{ExpLN}, the mixture Exponential-Lognormal model has been proposed to describe the irradiance fluctuations due to air bubbles in both fresh and salty waters in UWOC channels.
However, the mathematical form of Lognormal-based
distributions is not always convenient for analytic calculations.
Furthermore, the design and the performance analysis of such
systems is much more challenging. Indeed, the application
of the Exponential-Lognormal in UWOC channels makes it
very hard to obtain closed-form and easy-to-use expressions
for important performance metrics such as the outage probability
and the average bit-error rate (BER). The mathematical
intractability of the Lognormal-based model becomes more
evident when we know that the assessment of BER is based
on numerical methods, as closed-form analytical expressions
are not available for this model. In \cite{Oubei:17sg}, the Weibull distribution was used to characterize fluctuations of laser beam intensity in underwater caused by salinity gradient. The Generalized Gamma distribution (GGD) was proposed to accurately describe both non turbulent thermally uniform and gradient based underwater wireless optical channels in \cite{Oubei:17}.
In \cite{ConfUWOC}, the mixture Exponential-Gamma distribution was presented to characterize optical signal irradiance fluctuations in underwater channel in the presence of air bubbles for uniform temperature.
The
presence of the Lognormal distribution or the
Gamma distribution agrees with previous studies suggesting
its use to model underwater optical channels. The Exponential
distribution, is however, less common. As shown in \cite{ConfUWOC}, it is
 used to model the loss in the received energy caused by
air bubbles. Therefore, typical single-lobe distributions cannot
appropriately fit the measured data in the presence of air
bubbles, and a two-lobe statistical model is required to predict
the statistical behavior of UWOC turbulence-induced fading in
all regions of the scintillation index.
In \cite{JamaliModel}, different statistical distributions have been proposed to describe fading in UWOC channels under various conditions.

In \cite{UWOCtcom}, we have proposed a new statistical model to characterize turbulence-induced fading in UWOC channels in the presence of both air bubbles and temperature gradients for fresh and salty waters, based on an experiment conducted in an indoor laboratory setup. In fact, there are many sources of temperature gradient in the oceans. Influxes of glacial fresh water, extratropical cyclones, and ocean currents such as Labrador and Gulf Stream \cite{Siedler} are a few examples of temperature-induced turbulent UWOC channels.
In this model, the channel irradiance fluctuations are characterized by the mixture Exponential-Generalized Gamma (EGG) distribution, which is a weighted sum
of the Exponential and Generalized Gamma distributions \cite{UWOCtcom}. We used the expectation maximization
(EM) algorithm to obtain the maximum likelihood parameter
estimation of the new model.
We have demonstrated that this
model perfectly matches the measured data, collected under different lab channel scenarios ranging from weak to strong
turbulence conditions, for both salty as well as fresh waters.
A comparison with the Exponential-Lognormal model has been also
performed, and we have shown that the proposed distribution
can be applicable under various
conditions of irradiance fluctuations considered in the experimental setup in \cite{UWOCtcom}, providing a better fit to the measured data as well. Furthermore, this model being
simpler and analytically more tractable than the Exponential-Lognormal model, is more convenient for
performance analysis and design of UWOC systems.
Moreover, based on reference \cite{Oubei:17sg} where the Weibull distribution which is a special case of the Generalized Gamma distribution was used to fit irradiance fluctuations data due to underwater salinity gradient, the EGG model is expected to accurately capture a combination of air bubbles, gradient of temperature, and gradient of salinity fluctuations, making it a unified model that can address the statistics of optical beam irradiance fluctuations in a wide variety of turbulent underwater wireless optical channels. In addition, when the water temperature is uniform, the received intensity of
the laser beam is best described by the mixture Exponential-Gamma distribution which is a special case of the EGG model.

Dual-hop relaying, where an intermediate terminal relays the signal from the source terminal to
the destination terminal, can be used over UWOC links to mitigate scattering, absorption and turbulence-induced fading and extend
the viable communication range. This is due to the fact that these impairing
factors increase rapidly with distance \cite{JamaliRelay}. Therefore, dividing the long communication distance to shorter ones by means of intermediate relays
is an efficient technique to expand the coverage of underwater optical wireless sensor networks (UOWSNs) \cite{Celiksensor} with low power requirements, increase the reliability of the UWOC link, and
offer high data-rate at the end-to-end communication. In \cite{JamaliMultihop},
the average BER of point-to-point multi-hop UWOC systems is investigated. The BER of
relay-assisted underwater wireless optical code-division multiple
access (OCDMA) networks over turbulent channels have been addressed in \cite{JamaliRelay}. The authors in \cite{Celiksensor} investigated the connectivity of UOWSNs and its impacts on the network localization performance. The end-to-end performance of multi-hop underwater wireless optical networks using amplify-and-forward (AF) and the more complex decode-and-forward (DF)
relaying schemes has been investigated in \cite{Celik2018}. Indeed, DF systems use more complex relays that fully decode the source-relay signal
and retransmit the decoded version into the relay-destination hop. On the other hand, AF systems
just amplify the incoming signal and forward it to the second hop without performing any kind
of decoding and can be classified into two categories, namely, fixed-gain relays and channel state information (CSI)-assisted relays. AF systems using CSI-assisted relays need instantaneous CSI of the first hop to control the relay gain resulting in a signal with fixed power at the relay output. However, AF systems employing fixed gain relays do not require the knowledge of the instantaneous CSI of the first hop and
use amplifiers with fixed gains at the relays, and as a result the power of the retransmitted signal is variable. Although AF systems equipped with fixed gain relays provide lower performance as compared to systems using CSI-assisted relays, they have the advantage of being less complex and easy to deploy which make them attractive from a practical standpoint \cite{Hasna}.

However, to the best of authors' knowledge, this is the first closed-form performance analysis of dual-hop UWOC systems using AF fixed gain relaying over the newly proposed EGG fading model that includes several statical models as special cases \cite{UWOCtcom}, where each UWOC link operates under either the heterodyne or the intensity modulation with direct detection (IM/DD) in the presence of air bubbles and gradient of temperature fluctuations. We propose a novel mathematical framework to derive exact closed-form expressions for
the outage probability, the average BER for a variety of modulation schemes, and the ergodic capacity, while not making any assumptions in our derivations, in terms of the bivariate Fox's H function. Moreover, our performance study provides a generalized framework for several fading channels such as Generalized Gamma \cite{Oubei:17sg} and Weibull \cite{Oubei:17} distributions that we have proposed to characterize UWOC channel fading due to temperature-induced turbulence and salinity-induced turbulence, respectively.
Furthermore, we present new and very tight asymptotic expressions for the outage probability and the average BER
in the high SNR regime in terms of simple elementary functions. Capitalizing on these asymptotic results, we derive the diversity gain of the dual-hop UWOC system under study.

The remainder of this paper is organized as follows. In Section II, we present the channel and communication system model. We derive the cumulative distribution function (CDF), the probability density function (PDF), and the moments of the end-to-end signal-to-noise ratio (SNR) of dual-hop UWOC systems in Section III. Capitalizing on these results,
we present closed-form expressions of the outage probability, the average BER for a variety of modulation schemes, and the ergodic capacity along with the asymptotic analysis at high SNR regime in Section IV.
Section V presents some numerical and simulation results to illustrate the mathematical formalism presented in Sections III and IV. Finally, some concluding remarks are drawn in Section VI.

\section{System and Channel Models}
We consider a dual-hop UWOC system, where the source node S and the destination node
D are communicating through the help of an intermediate relay node R which relays the information signal from S to D, acting as a non-regenerative fixed gain relay.
We assume that there is no direct link between nodes S and D due to the unsatisfactory
quality of the channel between them, and the communication can be achieved only through the relay.
Moreover, the optical beam propagating through the UWOC channel is significantly affected by the scattering and absorption effects in addition to the underwater optical turbulence caused by air bubbles and gradient of temperature. More specifically, scattering and absorption effects attenuate the mean irradiance
of the light beam and result in path loss of the UWOC channel. On the other hand, turbulence results in fluctuations (scintillations) of the received signal and may lead to link outage which ultimately degrades the performance of UWOC channels \cite{advancedfso,Simodetection}. Under this combined effect of absorption and scattering as well as optical turbulence, the normalized channel fading is appropriately characterized by the mixture EGG model, based on experimental measured data \cite{UWOCtcom}.
Therefore, the two UWOC hops (i.e. S-R and R-D) are subject to independent but not necessarily identically distributed mixture EGG distribution \cite{UWOCtcom}
\begin{align}
f_{I_i}(I_i)= \frac{\omega_i}{\lambda_i}\exp\left(-\frac{I_i}{\lambda_i}\right)+(1-\omega_i)\frac{c_i\,I_i^{a_i c_i-1}}{b_i^{a_i c_i}\Gamma(a_i)}\exp\left(-\left(\frac{I_i}{b_i}\right)^{c_i}\right),\quad i=1,2,
\label{newmodel}
\end{align}
where $\omega_i$ is the mixture weight or mixture coefficient of the distributions satisfying $0<\omega_i<1$, $\lambda_i$ is the parameter associated with the Exponential distribution, $a_i,b_i$ and $c_i$ are the parameters of the Generalized Gamma distribution, and $\Gamma(.)$ denotes the Gamma function for $i=1, 2$.

The scintillation index for each UWOC link $\sigma_{I_i}^2$, defined as the normalized variance of the irradiance fluctuations is given by \cite[Eq.(6)]{UWOCtcom}
\begin{align}\label{SInewmodel}
\sigma_{I_i}^2=2\omega_i\lambda_i^2+(1-\omega_i)b_i^2\frac{\Gamma(a_i+\frac{2}{c_i})}{\Gamma(a_i)}-1.
\end{align}

Considering both types of detection techniques (IM/DD as well as heterodyne detection), the PDF of the
instantaneous SNR at the $i$-th hop $\gamma_i$, defined as $\gamma_i=(\eta I_i)^{r_i}/N_{0_{i}}$, can be derived from (\ref{newmodel}) as \cite[Eq.(21)]{UWOCtcom}
\begin{align}\label{SNRPDFHop}
f_{\gamma_i}(\gamma_i)=\frac{\omega_i}{r_i\,\gamma_i} \,{\rm{G}}_{0,1}^{1,0}\left[ \frac{1}{\lambda_i}\left ( \frac{\gamma_i}{\mu_{r_i}} \right )^{\frac{1}{r_i}}\left| \begin{matrix} {-} \\ {1} \\ \end{matrix} \right. \right]
+\frac{c_i(1-\omega_i)}{r_i \Gamma(a_i)\gamma_i}{\rm{G}}_{0,1}^{1,0}\left[\frac{1}{b_i^{c_i}} \left ( \frac{\gamma_i}{\mu_{r_i}} \right )^{\frac{c_i}{r_i}}\left| \begin{matrix} {-} \\ {a_i} \\ \end{matrix} \right. \right],
\end{align}
where $\eta$ stands for the effective photoelectric conversion ratio, $N_{0_i}$ is the power of the
additive white Gaussian noise (AWGN) at the $i$-th hop, $r_i$ represents the type of
detection being employed at each hop (i.e. $r_i=1$ is associated with heterodyne
detection and $r_i=2$ is associated with IM/DD), ${\rm{G}}_{\cdot,\cdot}^{\cdot,\cdot}(\cdot)$ is the Meijer's G function \cite[Eq.(9.301)]{Tableofintegrals}, and $\mu_{r_i}$ denotes to the average electrical SNR of the $i$-th hop for $i=1, 2$.
In particular, for $r_i=1$,
\begin{align}
\mu_{1_i}=\mu_{{\rm{heterodyne}}_i}=\mathbb{E}[\gamma_i]=\overline{\gamma}_i,
\end{align}
and for $r_i=2$,
\begin{align}
\mu_{2_i}=\mu_{{\rm{IM/DD}}_i}=
\frac{\bar{\gamma_i}}{2 \omega_i \lambda_i^2+b_i^2(1-\omega_i)\Gamma\left ( a_i+2/c_i \right )/\Gamma(a_i)}.
\end{align}\noindent
Moreover, the CDF of $\gamma_i$ can be expressed as \cite[Eq.(22)]{UWOCtcom}
\begin{align}\label{SNCDFHop}
F_{\gamma_i}(\gamma_i)&=\omega_i \,{\rm{G}}_{1,2}^{1,1}\left[ \frac{1}{\lambda_i}\left ( \frac{\gamma_i}{\mu_{r_i}} \right )^{\frac{1}{r_i}}\left| \begin{matrix} {1} \\ {1,0} \\ \end{matrix} \right. \right]
+\frac{(1-\omega_i)}{\Gamma(a_i)}{\rm{G}}_{1,2}^{1,1}\left[ \frac{1}{b_i^{c_i}}\left ( \frac{\gamma_i}{\mu_{r_i}} \right )^{\frac{c_i}{r_i}}\left| \begin{matrix} {1} \\ {a_i,0} \\ \end{matrix} \right. \right].
\end{align}
It is worthy to mention that in the special case of thermally uniform UWOC channels, the EGG model simplifies to the Exponential-Gamma model and therefore, the CDF expression can be simplified by setting $c_i=1$ in (\ref{SNCDFHop}).

The end-to-end instantaneous SNR of dual-hop UWOC systems
with fixed gain relays that introduce a fixed gain to the received signal regardless of the fading amplitude on the first hop and consequently result in a signal with variable power at the output of the relay can be written under the assumption of negligible saturation as \cite{Hasna,Yacoub2008,Sonia2009,Peppas2010,Ines2010}
\begin{align}\label{overallSNR}
\gamma=\frac{\gamma_1 \gamma_2}{\gamma_2+C},
\end{align}
where $C$ is a constant inversely proportional to the squared relay's gain such that $C=1/(G^2 N_{0_1})$, $G$ represents the relay gain established in the connection, $N_{0_1}$ stands for the power of the
additive white Gaussian noise (AWGN) at the first hop, and $\gamma_i$ represents the instantaneous SNR for the $i$-th
hop for $i=1, 2$, with the PDF and CDF given by (\ref{SNRPDFHop}) and (\ref{SNCDFHop}), respectively.
More specifically, we consider a semi-blind fixed gain relaying system that benefits from the knowledge
of the first hop's average fading power where the fixed
gain is set equal to the average of the CSI-assisted gain \cite{Hasna,Yacoub2008,Sonia2009,Peppas2010,Ines2010}, that is,
\begin{align}
G^2=\mathbb{E}\left [ \frac{1}{ (\eta I)^r +N_{0_1}} \right ],
\end{align}
which can be obtained in closed-form by using \cite[Eq.~(8.4.2/5)]{PrudinkovVol3} then \cite[Eq.(2.9.1)]{HTranforms}, and applying \cite[Eq.~(2.25.1/1)]{PrudinkovVol3} as
\begin{align}\label{Gsq}
G^2=\frac{1}{N_{0_1}}\left (\omega_1{\rm{H}}_{1,2}^{2,1}\left[\frac{1 }{\lambda_1^{r_1}\mu_{r_{1 }}}\left| \begin{matrix} {(1,1)} \\ {(1,r_1)(1,1)} \\ \end{matrix} \right. \right]+\frac{(1-\omega_1)}{\Gamma(a_1)}{\rm{H}}_{1,2}^{2,1}\left[\frac{1 }{b_1^{r_1}\mu_{r_{1 }}}\left| \begin{matrix} {(1,1)} \\ {\left(a_1,\frac{r_1}{c_1}\right)(1,1)} \\ \end{matrix} \right. \right]  \right ).
\end{align}
Finally, the parameter $C$ can be easily derived from (\ref{Gsq}) as $C=1/(G^2 N_{0_1})$.
It is important to mention that semi-blind relays are more attractive than CSI-assisted relays from a practical standpoint as they offer simplicity and ease of deployment. This is due to the fact that such systems do not require a continuous monitoring of the channel for its instantaneous knowledge, as compared to the CSI-assisted relaying case \cite{Hasna}.

\section{End-to-End SNR Statistics}
This section derives new closed-form expressions for the end-to-end SNR statistics of the dual-hop UWOC fixed gain relaying system that accounts for air bubbles and temperature gradients for fresh and salty waters, under both heterodyne detection and IM/DD techniques. A tractable and very tight asymptotic approximation for the CDF of the end-to-end SNR is also provided and in sequel the diversity order of the system is presented.

\begin{theorem}(Cumulative Distribution Function). The CDF of the end-to-end SNR $\gamma$ defined in (\ref{overallSNR})
can be obtained in exact closed-form by
\begin{align}\label{SNRCDF}
\nonumber F_\gamma(\gamma)&=1-\omega_1 \omega_2{\rm{H}}_{1,0:1,1:0,2}^{0,1:0,1:2,0}\begin{bmatrix}
\begin{matrix}
(1;1,1)\\--
\end{matrix}
\Bigg|\begin{matrix}
(0,r_1)\\(0,1)
\end{matrix}
\Bigg|\begin{matrix}
--\\(0,1)(1,r_2)
\end{matrix}
\Bigg|
\frac{\lambda_1^{r_1} \mu_{r_1}}{\gamma},\frac{C}{\lambda_2^{r_2} \mu_{r_2}}
\end{bmatrix}\\
\nonumber & -\frac{\omega_1 \left (1-\omega_2  \right )}{\Gamma(a_2)}
{\rm{H}}_{1,0:1,1:0,2}^{0,1:0,1:2,0}\begin{bmatrix}
\begin{matrix}
(1;1,1)\\--
\end{matrix}
\Bigg|\begin{matrix}
(0,r_1)\\(0,1)
\end{matrix}
\Bigg|\begin{matrix}
--\\(0,1)\left(a_2,\frac{r_2}{c_2}\right)
\end{matrix}
\Bigg|
\frac{\lambda_1^{r_1} \mu_{r_1}}{\gamma},\frac{C}{b_2^{r_2} \mu_{r_2}}
\end{bmatrix}\\
\nonumber & - \frac{\omega_2 \left (1-\omega_1  \right )}{\Gamma(a_1)}
{\rm{H}}_{1,0:1,1:0,2}^{0,1:0,1:2,0}\begin{bmatrix}
\begin{matrix}
(1;1,1)\\--
\end{matrix}
\Bigg|\begin{matrix}
\left (1-a_1,\frac{r_1}{c_1}  \right )\\(0,1)
\end{matrix}
\Bigg|\begin{matrix}
--\\(0,1)(1,r_2)
\end{matrix}
\Bigg|
\frac{b_1^{r_1} \mu_{r_1}}{\gamma},\frac{C}{\lambda_2^{r_2} \mu_{r_2}}
\end{bmatrix}\\
& - \frac{\left (1-\omega_1  \right )\left ( 1-\omega_2  \right )}{\Gamma(a_1)\Gamma(a_2)}{\rm{H}}_{1,0:1,1:0,2}^{0,1:0,1:2,0}\begin{bmatrix}
\begin{matrix}
(1;1,1)\\--
\end{matrix}
\Bigg|\begin{matrix}
\left (1-a_1,\frac{r_1}{c_1}  \right )\\(0,1)
\end{matrix}
\Bigg|\begin{matrix}
--\\(0,1)\left (a_2,\frac{r_2}{c_2}  \right )
\end{matrix}
\Bigg|
\frac{b_1^{r_1} \mu_{r_1}}{\gamma},\frac{C}{b_2^{r_2} \mu_{r_2}}
\end{bmatrix},
\end{align}
where ${\rm{H}}_{\cdot,\cdot:\cdot,\cdot:\cdot,\cdot}^{\cdot,\cdot:\cdot,\cdot:\cdot,\cdot}[\cdot]$ stands for the Fox's H function of two variables \cite{HFoxIntegrals}, known also as the bivariate Fox's H function, with a MATLAB
implementation presented in \cite{arxivFox,MutivariateHcode1,MutivariateHcode2}.
\end{theorem}
\begin{proof}
See Appendix A.
\end{proof}
It is worth noting that this closed-form result for the CDF is important and particularly useful to evaluate the outage probability performance of the dual-hop UWOC system as will be shown in the next section of this work. In addition, to obtain more engineering insights on the performance of the dual-hop UWOC system under study, we elaborate further
on the asymptotic analysis at high SNR regime in the following corollary.
\begin{corollary}
Assume that $\mu_{r_{1}},\mu_{r_{2}} \to \infty$. Then, the CDF in (\ref{SNRCDF}) can be expressed asymptotically in the high SNR regime, in terms of simple elementary functions as

\begin{align}\label{CDFHighSNR}
\nonumber  F_{\gamma}(\gamma)&\underset{\mu_{r_{1}},\mu_{r_{2}}\to \infty}{\mathop{\approx }}
\omega_1 \left (\frac{\gamma}{\lambda_1^{r_1}\mu_{r_{1 }}}  \right )^{\frac{1}{r_1}}
+\frac{(1-\omega_1)}{\Gamma(a_1+1)}\left ( \frac{\gamma}{b_1^{r_1}\mu_{r_{1}}} \right )^{\frac{a_1 c_1}{r_1}}+\frac{\omega_2(1-\omega_1)}{\Gamma(a_1)}\Gamma\left ( a_1-\frac{r_1}{c_1 r_2} \right )\\
\nonumber & \times \left ( \frac{C \gamma}{b_1^{r_1}\lambda_2^{r_2}\mu_{r_{1}}\mu_{r_{2}}} \right )^{\frac{1}{r_2}}+\frac{(1-\omega_2)}{\Gamma(a_2+1)}\left ( \frac{C \gamma}{b_2^{r_2}\mu_{r_{1 }}\mu_{r_{2 }}} \right )^{\frac{a_2 c_2}{r_2}}
\left (\frac{\omega_1}{\lambda_1^{\frac{r_1 a_2 c_2}{r_2}}}\Gamma\left ( 1-\frac{r_1 a_2 c_2}{r_2} \right )\right.\\
& \left. +\frac{(1-\omega_1)}{\Gamma(a_1)b_1^{\frac{r_1 a_2 c_2}{r_2}}}\Gamma\left ( a_1-\frac{r_1 a_2 c_2}{c_1 r_2} \right ) \right ).
\end{align}
\end{corollary}
\begin{proof}
See Appendix B.
\end{proof}
It is important to mention that the asymptotic expression of the CDF given in (\ref{CDFHighSNR}) includes only summations of basic elementary functions such as the Gamma function, as compared to the exact expression of the CDF obtained in terms of the bivariate Fox's H function in (\ref{SNRCDF}), which is a quite complex function and not a standard built-in function in most of the well-known mathematical software tools such as MATHEMATICA and MATLAB. Moreover,
the expression in (\ref{CDFHighSNR}) is simpler and much more analytically tractable than (\ref{SNRCDF}), and more importantly, is very accurate and converges perfectly to the exact result in (\ref{SNRCDF}) at high SNR regime, which is illustrated in section V. Furthermore, this tractable result is of particular importance when it comes to finding the diversity order of the dual-hop UWOC system under study that depends on the type of receiver
detection being used in each hop (i.e. $r_1$ and $r_2$), and the two
UWOC hop's turbulence parameters (i.e. $a_1$, $c_1$, $a_2$, and $c_2$), that is,
\begin{align}
G_d=\min\left( \frac{1}{r_1}, \frac{a_1 c_1}{r_1} ,\frac{2}{r_2},\frac{2 a_2 c_2}{r_2} \right).
\end{align}
\begin{theorem}(Probability Density Function). By taking the derivative of (\ref{SNRCDF}) with respect to $\gamma$, the PDF of the end-to-end SNR can be shown to be given in closed-form by
\begin{align}\label{SNRPDF}
\nonumber &f_\gamma(\gamma)=\frac{\omega_1 \omega_2}{\gamma}{\rm{H}}_{1,0:1,1:0,2}^{0,1:0,1:2,0}\begin{bmatrix}
\begin{matrix}
(1;1,1)\\--
\end{matrix}
\Bigg|\begin{matrix}
(0,r_1)\\(1,1)
\end{matrix}
\Bigg|\begin{matrix}
--\\(0,1)(1,r_2)
\end{matrix}
\Bigg|
\frac{\lambda_1^{r_1} \mu_{r_1}}{\gamma},\frac{C}{\lambda_2^{r_2} \mu_{r_2}}
\end{bmatrix}\\
\nonumber &+\frac{\omega_1 \left (1-\omega_2  \right )}{\Gamma(a_2)\gamma}
{\rm{H}}_{1,0:1,1:0,2}^{0,1:0,1:2,0}\begin{bmatrix}
\begin{matrix}
(1;1,1)\\--
\end{matrix}
\Bigg|\begin{matrix}
(0,r_1)\\(1,1)
\end{matrix}
\Bigg|\begin{matrix}
--\\(0,1)\left(a_2,\frac{r_2}{c_2}\right)
\end{matrix}
\Bigg|
\frac{\lambda_1^{r_1} \mu_{r_1}}{\gamma},\frac{C}{b_2^{r_2} \mu_{r_2}}
\end{bmatrix}\\
\nonumber &+\frac{\omega_2\left ( 1-\omega_1  \right )}{\Gamma(a_1)\gamma}{\rm{H}}_{1,0:1,1:0,2}^{0,1:0,1:2,0}\begin{bmatrix}
\begin{matrix}
(1;1,1)\\--
\end{matrix}
\Bigg|\begin{matrix}
\left(1-a_1,\frac{r_1}{c_1}\right)\\(1,1)
\end{matrix}
\Bigg|\begin{matrix}
--\\(0,1)(1,r_2)
\end{matrix}
\Bigg|
\frac{b_1^{r_1} \mu_{r_1}}{\gamma},\frac{C}{\lambda_2^{r_2} \mu_{r_2}}
\end{bmatrix}\\
&+\frac{\left (1-\omega_1  \right ) \left (1-\omega_2  \right )}{\Gamma(a_1)\Gamma(a_2)\gamma}
{\rm{H}}_{1,0:1,1:0,2}^{0,1:0,1:2,0}\begin{bmatrix}
\begin{matrix}
(1;1,1)\\--
\end{matrix}
\Bigg|\begin{matrix}
\left(1-a_1,\frac{r_1}{c_1}\right)\\(1,1)
\end{matrix}
\Bigg|\begin{matrix}
--\\(0,1)\left(a_2,\frac{r_2}{c_2}\right)
\end{matrix}
\Bigg|
\frac{b_1^{r_1} \mu_{r_1}}{\gamma},\frac{C}{b_2^{r_2} \mu_{r_2}}
\end{bmatrix}.
\end{align}
\end{theorem}
\begin{proof}
See Appendix C.
\end{proof}
\begin{theorem}(Moments).
The $n$th moments of $\gamma$ defined as $\mathbb{E}[\gamma^n]=\int_{0}^{\infty}\gamma^n\,f_\gamma(\gamma)\,d\gamma$, may be obtained in closed-form in terms of the Fox's function ${\rm{H}}_{\cdot,\cdot}^{\cdot,\cdot}[\cdot]$, which has an efficient MATHEMATICA$^{\textregistered}$ implementation in \cite{FerkanHFox}, as
\begin{align}\label{moments}
\nonumber & \mathbb{E}[\gamma^n]=\left (\frac{\omega_1 \omega_2 }{\Gamma(n)}\Gamma(1+r_1n) \left ( \lambda_1^{r_1}\mu_{r_{1}} \right )^n
+\frac{\omega_2 \left (1-\omega_1  \right )\left ( b_1^{r_1}\mu_{r_{1}} \right )^n }{\Gamma(a_1)\Gamma(n)} \Gamma\left (a_1+\frac{r_1n}{c_1} \right )  \right )\\
\nonumber &\times {\rm{H}}_{1,2}^{2,1}\left[\frac{C  }{\lambda_2^{r_2}\mu_{r_{2 }}}\left| \begin{matrix} {(1-n,1)} \\ {(0,1)(1,r_{2})} \\ \end{matrix} \right. \right]
+\left (\frac{\omega_1 \left (1-\omega_2  \right ) \Gamma(1+r_1n)\left ( \lambda_1^{r_1}\mu_{r_{1}} \right )^n }{\Gamma(a_2)\Gamma(n)}\right.\\
&\left.+\frac{ \left (1-\omega_1  \right )\left (1-\omega_2  \right ) \left ( b_1^{r_1}\mu_{r_{1}} \right )^n}{\Gamma(a_1)\Gamma(a_2)\Gamma(n)} \Gamma\left (a_1+\frac{r_1}{c_1}n  \right )  \right )
{\rm{H}}_{1,2}^{2,1}\left[\frac{C }{b_2^{r_2}\mu_{r_{2 }}}\left| \begin{matrix} {(1-n,1)} \\ {(0,1)\left (a_2,\frac{r_2}{c_2}  \right )} \\ \end{matrix} \right. \right].
\end{align}
\end{theorem}
\begin{proof}
See Appendix D.
\end{proof}
Note that the above result for the moments can be used for the calculation of the higher-order amount of fading which is an important performance measure defined in \cite{AFN} as $AF_\gamma^{(n)}=\mathbb{E}[\gamma^n]/\mathbb{E}[\gamma]^n-1$.

\section{Performance Metrics}
This section provides new analytical expressions for the exact and
asymptotic key performance metrics of the dual-hop UWOC system, in the presence of air bubbles and temperature gradients for IM/DD and heterodyne techniques.

\subsection{Outage Probability}
The outage probability is a fundamental performance measure of UWOC communication systems. It is encountered when the end-to-end SNR, $\gamma$, falls below a certain specified threshold $\gamma_{\rm{th}}$.
By setting $\gamma=\gamma_{\text{th}}$ in (\ref{SNRCDF}), the end-to-end outage probability of the dual-hop UWOC system in operation under both heterodyne detection as well as IM/DD can be easily obtained in exact closed-form as
\begin{align}\label{OP}
P_{\text{out}}(\gamma_{\text{th}})=F_\gamma(\gamma_{\text{th}}).
\end{align}

\subsection{Average Bit-Error Rate}
\subsubsection{Exact Analysis}
A generalized expression for the average BER for a variety of modulation schemes under both heterodyne and IM/DD techniques can be expressed as \cite{dualhopFSO}
\begin{align}\label{BERdef}
P_e=\frac{\delta}{2 \Gamma(p)}\sum_{k=1}^{n}\int_{0}^{\infty}\Gamma(p,q_k\,\gamma) f_\gamma(\gamma)\,d\gamma,
\end{align}
where $\Gamma(\cdot,\cdot)$ is the upper incomplete Gamma function \cite[Eq.(8.350/2)]{Tableofintegrals}, $n$, $\delta$, $p$, and $q_k$ vary depending on the modulation technique and the type of detection (i.e IM/DD or heterodyne detection) and are listed in \cite[Table III]{UWOCtcom}. It is important to mention here that for IM/DD technique, we investigate the average BER for on-off keying (OOK) modulation since it is the most commonly used intensity modulation technique in practical UWOC systems due to its simplicity and resilience to laser nonlinearity. For heterodyne detection and in addition to binary modulation schemes, we analyze the average BER for multilevel phase shift keying (MPSK) and quadrature amplitude (MQAM) that are commonly deployed in coherent systems.

\begin{theorem}(Average Bit-Error Rate).
By using (\ref{BERdef}), a unified expression for the average BER for all these modulation schemes of a dual-hop UWOC system operating under both IM/DD and heterodyne techniques can be derived in exact closed-form in terms of the bivariate Fox's H function as
\begin{align}\label{SNRBER}
\nonumber P_e&=\delta\sum_{k=1}^{n}\left\{\frac{1}{2}-\frac{\omega_1 \omega_2}{2\Gamma(p)}{\rm{H}}_{1,0:1,2:0,2}^{0,1:1,1:2,0}\begin{bmatrix}
\begin{matrix}
(1;1,1)\\--
\end{matrix}
\Bigg|\begin{matrix}
(0,r_1)\\(p,1)(0,1)
\end{matrix}
\Bigg|\begin{matrix}
--\\(0,1)(1,r_2)
\end{matrix}
\Bigg|
q_k\lambda_1^{r_1} \mu_{r_1},\frac{C}{\lambda_2^{r_2} \mu_{r_2}}
\end{bmatrix}\right.\\
\nonumber &\left.-\frac{\omega_1 (1-\omega_2)}{2\Gamma(a_2)\Gamma(p)}{\rm{H}}_{1,0:1,2:0,2}^{0,1:1,1:2,0}\begin{bmatrix}
\begin{matrix}
(1;1,1)\\--
\end{matrix}
\Bigg|\begin{matrix}
(0,r_1)\\(p,1)(0,1)
\end{matrix}
\Bigg|\begin{matrix}
--\\(0,1)\left(a_2,\frac{r_2}{c_2}\right)
\end{matrix}
\Bigg|
q_k\lambda_1^{r_1} \mu_{r_1},\frac{C}{b_2^{r_2} \mu_{r_2}}
\end{bmatrix}\right.\\
\nonumber &\left.-\frac{\omega_2 (1-\omega_1)}{2\Gamma(a_1)\Gamma(p)}{\rm{H}}_{1,0:1,2:0,2}^{0,1:1,1:2,0}\begin{bmatrix}
\begin{matrix}
(1;1,1)\\--
\end{matrix}
\Bigg|\begin{matrix}
\left(1-a_1,\frac{r_1}{c_1}\right)\\(p,1)(0,1)
\end{matrix}
\Bigg|\begin{matrix}
--\\(0,1)(1,r_2)
\end{matrix}
\Bigg|
q_kb_1^{r_1} \mu_{r_1},\frac{C}{\lambda_2^{r_2} \mu_{r_2}}
\end{bmatrix}\right.\\
&\left.-\frac{(1-\omega_1) (1-\omega_2)}{2\Gamma(a_1)\Gamma(a_2)\Gamma(p)}{\rm{H}}_{1,0:1,2:0,2}^{0,1:1,1:2,0}\begin{bmatrix}
\begin{matrix}
(1;1,1)\\--
\end{matrix}
\Bigg|\begin{matrix}
\left(1-a_1,\frac{r_1}{c_1}\right)\\(p,1)(0,1)
\end{matrix}
\Bigg|\begin{matrix}
--\\(0,1)\left(a_2,\frac{r_2}{c_2}\right)
\end{matrix}
\Bigg|
q_kb_1^{r_1} \mu_{r_1},\frac{C}{b_2^{r_2} \mu_{r_2}}
\end{bmatrix}\right\}.
\end{align}
\end{theorem}

\begin{proof}
See Appendix E.
\end{proof}
\subsubsection{High SNR Analysis}
The average BER expression in (\ref{BERdef}) may be re-written in terms of the CDF of $\gamma$ by using integration by parts as
\begin{align}\label{BERDEF2}
P_e= \frac{\delta }{2\Gamma(p)}\sum_{k=1}^{n}q_k^p\int_{0}^{\infty}\gamma^{p-1}e^{-q_k\gamma}F_\gamma(\gamma)\,d\gamma.
\end{align}
Substituting (\ref{CDFHighSNR}) into (\ref{BERDEF2}) then using \cite[Eq.(3.381/4)]{Tableofintegrals}, a very tight asymptotic expression of the average BER for a variety of modulation techniques can be obtained at high SNR in terms of basic elementary functions as shown in (\ref{BERHighSNR}).

\begin{align}\label{BERHighSNR}
\nonumber &P_e\underset{\mu_{r_{1}},\mu_{r_{2}}\to \infty}{\mathop{\approx }}
\frac{\delta \omega_1 \Gamma\left ( p+\frac{1}{r_1} \right )}{2\Gamma(p)} \sum_{k=1}^{n}\left (\frac{1}{q_k\lambda_1^{r_1}\mu_{r_{1 }}}  \right )^{\frac{1}{r_1}}
+\frac{\delta(1-\omega_1)\Gamma\left ( p+\frac{a_1 c_1}{r_1} \right )}{2\Gamma(a_1+1)\Gamma(p)}\sum_{k=1}^{n}
\left ( \frac{1}{q_k b_1^{r_1}\mu_{r_{1}}} \right )^{\frac{a_1 c_1}{r_1}}\\
\nonumber &+\frac{\delta\omega_2(1-\omega_1)\Gamma\left ( a_1-\frac{r_1}{c_1 r_2} \right )\Gamma\left ( p+\frac{1}{r_2} \right )}{2\Gamma(a_1)\Gamma(p)}\sum_{k=1}^{n} \left ( \frac{C }{q_k b_1^{r_1}\lambda_2^{r_2}\mu_{r_{1}}\mu_{r_{2}}} \right )^{\frac{1}{r_2}}+\frac{\delta (1-\omega_2)\Gamma\left ( p+\frac{a_2 c_2}{r_2} \right )}{2\Gamma(a_2+1)\Gamma(p)}\\
&\times\left (\frac{\omega_1}{\lambda_1^{\frac{r_1 a_2 c_2}{r_2}}}\Gamma\left ( 1-\frac{r_1 a_2 c_2}{r_2} \right )\right.
 \left. +\frac{(1-\omega_1)}{\Gamma(a_1)b_1^{\frac{r_1 a_2 c_2}{r_2}}}\Gamma\left ( a_1-\frac{r_1 a_2 c_2}{c_1 r_2} \right ) \right )
\sum_{k=1}^{n}\left ( \frac{C }{q_k b_2^{r_2}\mu_{r_{1 }}\mu_{r_{2 }}} \right )^{\frac{a_2 c_2}{r_2}}.
\end{align}\normalsize

\subsection{Ergodic Capacity}
The ergodic capacity of dual-hop UWOC communication systems can be given by \cite[Eq.(26)]{lapidoth}, \cite[Eq.(7.43)]{advancedfso}
\begin{align}\label{CAPDEF}
\overline{C}\triangleq \mathbb{E}[\ln(1+\tau\,\gamma)]=\int_{0}^{\infty}\ln(1+\tau\,\gamma)f_\gamma(\gamma)\,d\gamma,
\end{align}
where $\tau=e/(2\pi)$ is associated with the IM/DD technique (i.e. $r=2$) and $\tau=1$ is associated with the heterodyne detection technique (i.e. $r=1$). It is worthy to mention that the expression in (\ref{CAPDEF}) is exact when $r=1$ while it is a lower-bound when $r=2$.

\begin{theorem}(Ergodic Capacity). A unified expression for the ergodic capacity of AF fixed gain dual-hop UWOC systems over the mixture EGG fading channels can be derived in exact closed-form in terms of the bivariate Fox's H function as
\begin{align}\label{SNRCapacity}
\nonumber \overline{C}&=\omega_1 \omega_2{\rm{H}}_{1,0:2,2:0,2}^{0,1:1,2:2,0}\begin{bmatrix}
\begin{matrix}
(1;1,1)\\--
\end{matrix}
\Bigg|\begin{matrix}
(0,r_1)(1,1)\\(1,1)(0,1)
\end{matrix}
\Bigg|\begin{matrix}
--\\(0,1)(1,r_2)
\end{matrix}
\Bigg|
\tau\lambda_1^{r_1} \mu_{r_1},\frac{C}{\lambda_2^{r_2} \mu_{r_2}}
\end{bmatrix}\\
\nonumber &+\frac{\omega_1 \left (1-\omega_2  \right )}{\Gamma(a_2)}
{\rm{H}}_{1,0:2,2:0,2}^{0,1:1,2:2,0}\begin{bmatrix}
\begin{matrix}
(1;1,1)\\--
\end{matrix}
\Bigg|\begin{matrix}
(0,r_1)(1,1)\\(1,1)(0,1)
\end{matrix}
\Bigg|\begin{matrix}
--\\(0,1)\left(a_2,\frac{r_2}{c_2}\right)
\end{matrix}
\Bigg|
\tau\lambda_1^{r_1} \mu_{r_1},\frac{C}{b_2^{r_2} \mu_{r_2}}
\end{bmatrix}\\
\nonumber &+\frac{\omega_2 (1-\omega_1)}{\Gamma(a_1)}{\rm{H}}_{1,0:2,2:0,2}^{0,1:1,2:2,0}\begin{bmatrix}
\begin{matrix}
(1;1,1)\\--
\end{matrix}
\Bigg|\begin{matrix}
(1,1)\left ( 1-a_1,\frac{r_1}{c_1} \right )\\(1,1)(0,1)
\end{matrix}
\Bigg|\begin{matrix}
--\\(0,1)(1,r_2)
\end{matrix}
\Bigg|
\tau b_1^{r_1} \mu_{r_1},\frac{C}{\lambda_2^{r_2} \mu_{r_2}}
\end{bmatrix}\\
 &+\frac{ (1-\omega_1)(1-\omega_2)}{\Gamma(a_1)\Gamma(a_2)}{\rm{H}}_{1,0:2,2:0,2}^{0,1:1,2:2,0}\begin{bmatrix}
\begin{matrix}
(1;1,1)\\--
\end{matrix}
\Bigg|\begin{matrix}
(1,1)\left ( 1-a_1,\frac{r_1}{c_1} \right )\\(1,1)(0,1)
\end{matrix}
\Bigg|\begin{matrix}
--\\(0,1)\left ( a_2,\frac{r_2}{c_2} \right )
\end{matrix}
\Bigg|
\tau b_1^{r_1} \mu_{r_1},\frac{C}{b_2^{r_2} \mu_{r_2}}
\end{bmatrix}.
\end{align}
\end{theorem}

\begin{proof}
See Appendix F.
\end{proof}

\section{Numerical Results}
In this section, we provide some numerical results to illustrate the
outage probability, the average BER, and the ergodic capacity of the dual-hop UWOC system in the presence of air bubbles and temperature gradients for both fresh and salty waters under both IM/DD as well heterodyne techniques. Monte-Carlo simulations are also included to prove the correctness of the obtained results.
Unless otherwise
specified, we consider the case where the S-R and R-D hops are balanced, i.e. $\overline{\gamma}_1=\overline{\gamma}_2=\overline{\gamma}$. The parameter $C$ can be evaluated using (\ref{Gsq}).
The parameters $\omega_i$, $\lambda_i$, $a_i$, $b_i$ and $c_i$ with $i=1, 2$ for different levels of air bubbles under thermally uniform and gradient-based UWOC channels are obtained from laboratory experiments as listed in \cite[Table I]{UWOCtcom} and \cite[Table II]{UWOCtcom}, respectively. For readers clarification, to the best of our knowledge, a mathematical formulation for evaluating these parameters is not available in the open literature and can be considered as an open research challenge.
It is noteworthy that the experimental temperature gradient values reported in \cite{Oubei:17} are nearly 10 times higher than sea surface temperature (SST) of Gulf Stream and Kuroshio currents. According to \cite{William94}, the largest temperature difference between the reservoir of warm water at the ocean surface and the reservoir of cold water deeper in the ocean varies between 22 $^\circ C$ and 25 $^\circ C$. Therefore, the maximum temperature gradient level used in our experiment in \cite{Oubei:17}
mimics almost any realistic scenario encountered in the ocean, taking into account extreme conditions as well.
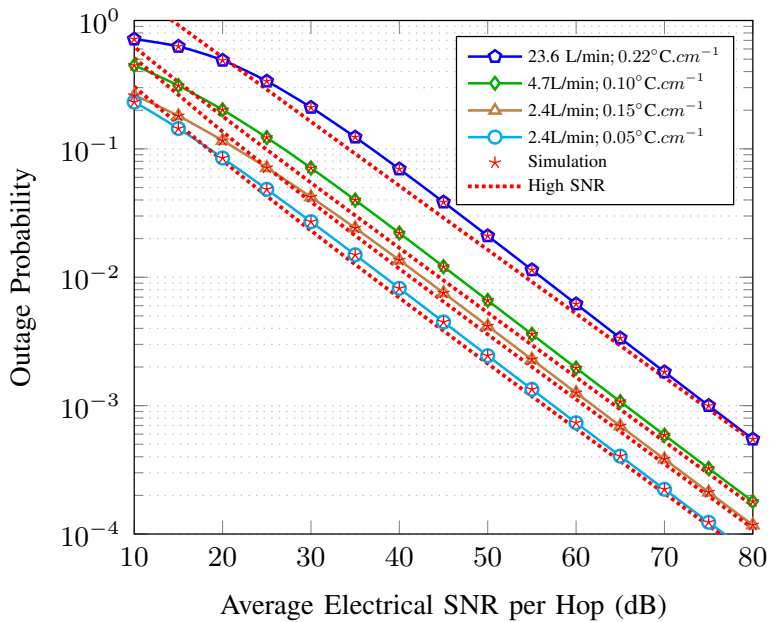
\begin{figure}[!h]
   \begin{center}
\begin{tikzpicture}[scale=1.2]
    \begin{axis}[font=\footnotesize,
      xmin=10,xmax=80,ymin=1e-04,ymax=1, ymode=log, xlabel=Average Electrical SNR per Hop (dB), ylabel= Outage Probability,
     legend style={nodes=right}, legend pos= north east,legend style={nodes={scale=0.63, transform shape}},
    xminorgrids,
    grid style={dotted},
    yminorgrids,]
    \addplot[smooth,blue,mark=pentagon*,mark options={solid},every mark/.append style={solid, fill=white}, thick] plot coordinates {
(0.000000,8.565300e-01)(5.000000,7.784690e-01)(10.000000,7.158760e-01)(15.000000,6.274630e-01)(20.000000,4.895210e-01)(25.000000,3.366800e-01)(30.000000,2.102920e-01)(35.000000,1.233430e-01)(40.000000,6.964320e-02)(45.000000,3.847320e-02)(50.000000,2.100780e-02)(55.000000,1.141070e-02)(60.000000,6.189870e-03)(65.000000,3.361610e-03)(70.000000,1.830340e-03)(75.000000,9.999140e-04)(80.000000,5.482290e-04)
};
    \addplot[smooth,
green!70!black,mark=diamond*,mark options={solid},every mark/.append style={solid, fill=white}, thick] plot coordinates {
(0.000000,6.928400e-01)(5.000000,5.949510e-01)(10.000000,4.462920e-01)(15.000000,3.108560e-01)(20.000000,2.015070e-01)(25.000000,1.224440e-01)(30.000000,7.098890e-02)(35.000000,3.994120e-02)(40.000000,2.208150e-02)(45.000000,1.209480e-02)(50.000000,6.598170e-03)(55.000000,3.597000e-03)(60.000000,1.963440e-03)(65.000000,1.074360e-03)(70.000000,5.896220e-04)(75.000000,3.246150e-04)(80.000000,1.792680e-04)
   };

    \addplot[smooth,brown,mark=triangle*,mark options={solid},every mark/.append style={solid, fill=white}, thick] plot coordinates {
(0.000000,9.999670e-01)(5.000000,3.177640e-01)(10.000000,2.597610e-01)(15.000000,1.805730e-01)(20.000000,1.160240e-01)(25.000000,7.125960e-02)(30.000000,4.213520e-02)(35.000000,2.417750e-02)(40.000000,1.359240e-02)(45.000000,7.547870e-03)(50.000000,4.164170e-03)(55.000000,2.291300e-03)(60.000000,1.260450e-03)(65.000000,6.941960e-04)(70.000000,3.830670e-04)(75.000000,2.118590e-04)(80.000000,1.174420e-04)   };
    \addplot[smooth,cyan,mark=*,mark options={solid},every mark/.append style={solid, fill=white}, thick] plot coordinates {
(0.000000,9.980230e-01)(5.000000,3.287650e-01)(10.000000,2.312130e-01)(15.000000,1.441700e-01)(20.000000,8.489230e-02)(25.000000,4.844890e-02)(30.000000,2.709550e-02)(35.000000,1.496020e-02)(40.000000,8.201450e-03)(45.000000,4.482740e-03)(50.000000,2.449570e-03)(55.000000,1.340500e-03)(60.000000,7.353420e-04)(65.000000,4.045290e-04)(70.000000,2.232000e-04)(75.000000,1.235020e-04)(80.000000,6.851640e-05)
};

      \addplot[smooth, mark=star,red,only marks] plot coordinates {
(0.000000,8.565300e-01)(5.000000,7.784690e-01)(10.000000,7.158760e-01)(15.000000,6.274630e-01)(20.000000,4.895210e-01)(25.000000,3.366800e-01)(30.000000,2.102920e-01)(35.000000,1.233430e-01)(40.000000,6.964320e-02)(45.000000,3.847320e-02)(50.000000,2.100780e-02)(55.000000,1.141070e-02)(60.000000,6.189870e-03)(65.000000,3.361610e-03)(70.000000,1.830340e-03)(75.000000,9.999140e-04)(80.000000,5.482290e-04)
   };
      \addplot[smooth, mark=star,red,only marks] plot coordinates {
(0.000000,6.928400e-01)(5.000000,5.949510e-01)(10.000000,4.462920e-01)(15.000000,3.108560e-01)(20.000000,2.015070e-01)(25.000000,1.224440e-01)(30.000000,7.098890e-02)(35.000000,3.994120e-02)(40.000000,2.208150e-02)(45.000000,1.209480e-02)(50.000000,6.598170e-03)(55.000000,3.597000e-03)(60.000000,1.963440e-03)(65.000000,1.074360e-03)(70.000000,5.896220e-04)(75.000000,3.246150e-04)(80.000000,1.792680e-04)
   };
      \addplot[smooth, mark=star,red,only marks] plot coordinates {
(0.000000,9.999670e-01)(5.000000,3.177640e-01)(10.000000,2.597610e-01)(15.000000,1.805730e-01)(20.000000,1.160240e-01)(25.000000,7.125960e-02)(30.000000,4.213520e-02)(35.000000,2.417750e-02)(40.000000,1.359240e-02)(45.000000,7.547870e-03)(50.000000,4.164170e-03)(55.000000,2.291300e-03)(60.000000,1.260450e-03)(65.000000,6.941960e-04)(70.000000,3.830670e-04)(75.000000,2.118590e-04)(80.000000,1.174420e-04)
};
      \addplot[smooth, mark=star,red,only marks] plot coordinates {
(0.000000,9.980230e-01)(5.000000,3.287650e-01)(10.000000,2.312130e-01)(15.000000,1.441700e-01)(20.000000,8.489230e-02)(25.000000,4.844890e-02)(30.000000,2.709550e-02)(35.000000,1.496020e-02)(40.000000,8.201450e-03)(45.000000,4.482740e-03)(50.000000,2.449570e-03)(55.000000,1.340500e-03)(60.000000,7.353420e-04)(65.000000,4.045290e-04)(70.000000,2.232000e-04)(75.000000,1.235020e-04)(80.000000,6.851640e-05)
 };

      \addplot[smooth,red,densely dotted, very  thick] plot coordinates {
(0.000000,5.048270e+00)(5.000000,2.894560e+00)(10.000000,1.636600e+00)(15.000000,9.203350e-01)(20.000000,5.167530e-01)(25.000000,2.901570e-01)(30.000000,1.630140e-01)(35.000000,9.164600e-02)(40.000000,5.155880e-02)(45.000000,2.902600e-02)(50.000000,1.635180e-02)(55.000000,9.218190e-03)(60.000000,5.200440e-03)(65.000000,2.936090e-03)(70.000000,1.659050e-03)(75.000000,9.382880e-04)(80.000000,5.311710e-04)
   };
      \addplot[smooth,red,densely dotted, very  thick] plot coordinates {
(0.000000,2.407900e+00)(5.000000,1.204860e+00)(10.000000,6.231190e-01)(15.000000,3.311950e-01)(20.000000,1.795040e-01)(25.000000,9.852480e-02)(30.000000,5.450830e-02)(35.000000,3.030920e-02)(40.000000,1.690910e-02)(45.000000,9.454480e-03)(50.000000,5.294520e-03)(55.000000,2.968200e-03)(60.000000,1.665340e-03)(65.000000,9.349050e-04)(70.000000,5.250710e-04)(75.000000,2.949890e-04)(80.000000,1.657670e-04)
   };
      \addplot[smooth,red,densely dotted, very  thick] plot coordinates {
(0.000000,-2.203200e+02)(5.000000,-7.577420e-01)(10.000000,5.103300e-01)(15.000000,2.635970e-01)(20.000000,1.360050e-01)(25.000000,7.192160e-02)(30.000000,3.876060e-02)(35.000000,2.116270e-02)(40.000000,1.165580e-02)(45.000000,6.457720e-03)(50.000000,3.592400e-03)(55.000000,2.004190e-03)(60.000000,1.120440e-03)(65.000000,6.273260e-04)(70.000000,3.516240e-04)(75.000000,1.972510e-04)(80.000000,1.107200e-04)
   };
      \addplot[smooth,red,densely dotted, very  thick] plot coordinates {
(0.000000,1.259470e+00)(5.000000,6.102260e-01)(10.000000,3.057700e-01)(15.000000,1.556750e-01)(20.000000,8.088130e-02)(25.000000,4.287240e-02)(30.000000,2.309250e-02)(35.000000,1.258560e-02)(40.000000,6.917200e-03)(45.000000,3.824790e-03)(50.000000,2.124120e-03)(55.000000,1.183410e-03)(60.000000,6.608650e-04)(65.000000,3.697000e-04)(70.000000,2.070870e-04)(75.000000,1.161120e-04)(80.000000,6.515130e-05)
   };
   \legend{23.6 L/min$;0.22^\circ$C.$cm^{-1}$,4.7L/min$;0.10^\circ$C.$cm^{-1}$,
   2.4L/min$;0.15^\circ$C.$cm^{-1}$,2.4L/min$;0.05^\circ$C.$cm^{-1}$, Simulation,,,,High SNR};

%
  \end{axis}
    \end{tikzpicture}
   \caption{Outage probability for different levels of air bubbles and gradient temperatures in the case of IM/DD technique along with the asymptotic results at high SNR for $\gamma_{\rm{th}}=0$ dB.}
      \label{fig:OP1new}
          \end{center}
\end{figure}

Fig.~\ref{fig:OP1new} demonstrates the impact of the temperature gradient as well as the air bubbles on
the end-to-end outage probability of the dual-hop UWOC system operating under the IM/DD technique.
 Clearly, we can observe from Fig.~\ref{fig:OP1new} that the analytical results provide a perfect match to the MATLAB simulated results proving the accuracy of our derivations. As expected, it can be shown from Fig.~\ref{fig:OP1new} that the higher is the level of the air bubbles and/or the temperature gradient, the higher is the value of the scintillation index and therefore, the stronger is the turbulence leading to a performance degradation.
For instance, at SNR=30 dB, $P_{\text{out}}=2.71 \times 10^{-2}$ for a temperature gradient equal to $0.05^\circ$C.$cm^{-1}$ and $\sigma_I^2=0.1484$ and it increases to $P_{\text{out}}=4.21 \times 10^{-2}$ for a temperature gradient of $0.15^\circ$C.$cm^{-1}$ and $\sigma_I^2=0.1915$, for a fixed bubbles level (BL), i.e. BL=2.4 L/min. This demonstrates the role of the temperature gradient in introducing severe irradiance fluctuations and hence severe turbulence conditions.
The asymptotic results of the outage probability at high SNR values obtained by using (\ref{CDFHighSNR}) are also included in Fig.~\ref{fig:OP1new}. As clearly seen from this figure, the asymptotic results of the outage probability are in a perfect match with the analytical results in the high SNR regime. This justifies the accuracy and the tightness of the derived asymptotic expression in (\ref{CDFHighSNR}).
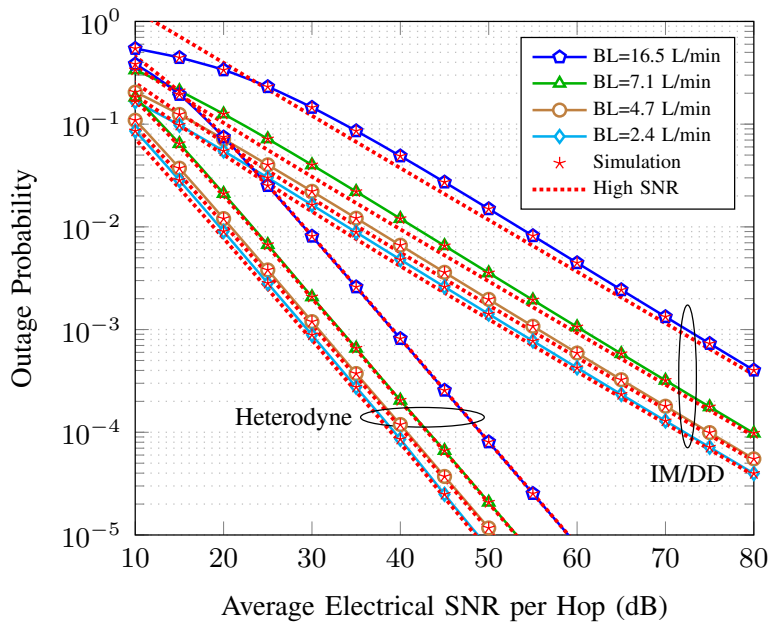
\begin{figure}[!h]
   \begin{center}
\begin{tikzpicture}[scale=1.2]
    \begin{axis}[xmin=10,xmax=80,font=\footnotesize,
      ymin=1e-05,ymax=1, ymode=log, xlabel= Average Electrical SNR per Hop (dB), ylabel= Outage Probability,
       legend style={nodes=right},legend style={nodes={scale=0.7, transform shape}},
    legend pos= north east,
    xminorgrids,
    grid style={dotted},
    yminorgrids,]
 \addplot[smooth,blue,mark=pentagon*,mark options={solid},every mark/.append style={solid, fill=white},thick] plot coordinates {
(0.000000,7.344340e-01)(5.000000,6.534960e-01)(10.000000,5.448840e-01)(15.000000,4.440380e-01)(20.000000,3.369560e-01)(25.000000,2.312290e-01)(30.000000,1.453700e-01)(35.000000,8.597610e-02)(40.000000,4.893840e-02)(45.000000,2.723670e-02)(50.000000,1.497320e-02)(55.000000,8.183190e-03)(60.000000,4.463870e-03)(65.000000,2.436380e-03)(70.000000,1.332750e-03)(75.000000,7.310500e-04)(80.000000,4.019820e-04)
};
 \addplot[smooth,green!70!black,mark=triangle*,mark options={solid},every mark/.append style={solid, fill=white},thick] plot coordinates {
(0.000000,6.497200e-01)(5.000000,4.998840e-01)(10.000000,3.367460e-01)(15.000000,2.110200e-01)(20.000000,1.254530e-01)(25.000000,7.178640e-02)(30.000000,4.006240e-02)(35.000000,2.203020e-02)(40.000000,1.202210e-02)(45.000000,6.541110e-03)(50.000000,3.558940e-03)(55.000000,1.939850e-03)(60.000000,1.060290e-03)(65.000000,5.814290e-04)(70.000000,3.199080e-04)(75.000000,1.765870e-04)(80.000000,9.777600e-05)
};
  \addplot[smooth,brown,mark=*,mark options={solid},every mark/.append style={solid, fill=white},thick] plot coordinates {
(0.000000,9.957820e-01)(5.000000,3.027220e-01)(10.000000,2.047360e-01)(15.000000,1.242840e-01)(20.000000,7.157770e-02)(25.000000,4.014050e-02)(30.000000,2.217070e-02)(35.000000,1.214060e-02)(40.000000,6.621480e-03)(45.000000,3.608210e-03)(50.000000,1.968480e-03)(55.000000,1.076460e-03)(60.000000,5.904210e-04)(65.000000,3.248710e-04)(70.000000,1.793190e-04)(75.000000,9.926950e-05)(80.000000,5.510000e-05)
};
  \addplot[smooth,cyan,mark=diamond*,mark options={solid},every mark/.append style={solid, fill=white},thick] plot coordinates {
(0.000000,9.943300e-01)(5.000000,2.521880e-01)(10.000000,1.652380e-01)(15.000000,9.770440e-02)(20.000000,5.483490e-02)(25.000000,3.008530e-02)(30.000000,1.635220e-02)(35.000000,8.858170e-03)(40.000000,4.798190e-03)(45.000000,2.603910e-03)(50.000000,1.417360e-03)(55.000000,7.742680e-04)(60.000000,4.245530e-04)(65.000000,2.336490e-04)(70.000000,1.290250e-04)(75.000000,7.146820e-05)(80.000000,3.969490e-05)
   };

 \addplot[smooth,blue,mark=pentagon*,mark options={solid},every mark/.append style={solid, fill=white},thick] plot coordinates {
(0.000000,7.441010e-01)(5.000000,5.890720e-01)(10.000000,3.840420e-01)(15.000000,1.942880e-01)(20.000000,7.476300e-02)(25.000000,2.519950e-02)(30.000000,8.095290e-03)(35.000000,2.597250e-03)(40.000000,8.159720e-04)(45.000000,2.562010e-04)(50.000000,8.047700e-05)(55.000000,2.529580e-05)(60.000000,7.956220e-06)(65.000000,2.503910e-06)(70.000000,7.884030e-07)(75.000000,2.483510e-07)(80.000000,7.826070e-08)

};
 \addplot[smooth,green!70!black,mark=triangle*,mark options={solid},every mark/.append style={solid, fill=white},thick] plot coordinates {
 (0.000000,9.976140e-01)(5.000000,4.401860e-01)(10.000000,1.826500e-01)(15.000000,6.457290e-02)(20.000000,2.108630e-02)(25.000000,6.707870e-03)(30.000000,2.077890e-03)(35.000000,6.567230e-04)(40.000000,2.045790e-04)(45.000000,6.636430e-05)(50.000000,2.079170e-05)(55.000000,6.522400e-06)(60.000000,2.048240e-06)(65.000000,6.437640e-07)(70.000000,2.024780e-07)(75.000000,6.372090e-08)(80.000000,2.006320e-08)

};
  \addplot[smooth,brown,mark=*,mark options={solid},every mark/.append style={solid, fill=white},thick] plot coordinates {
(0.000000,9.997580e-01)(5.000000,2.733500e-01)(10.000000,1.091070e-01)(15.000000,3.725650e-02)(20.000000,1.203510e-02)(25.000000,3.805430e-03)(30.000000,1.194820e-03)(35.000000,3.732030e-04)(40.000000,1.189560e-04)(45.000000,3.713930e-05)(50.000000,1.169340e-05)(55.000000,3.534850e-06)(60.000000,1.132660e-06)(65.000000,3.666460e-07)(70.000000,1.188100e-07)(75.000000,3.681580e-08)(80.000000,1.173560e-08)

};
  \addplot[smooth,cyan,mark=diamond*,mark options={solid},every mark/.append style={solid, fill=white},thick] plot coordinates {
(0.000000,9.995120e-01)(5.000000,2.265910e-01)(10.000000,8.572030e-02)(15.000000,2.838630e-02)(20.000000,9.022110e-03)(25.000000,2.823970e-03)(30.000000,8.863630e-04)(35.000000,2.751750e-04)(40.000000,8.563210e-05)(45.000000,2.489640e-05)(50.000000,7.560020e-06)(55.000000,2.897050e-06)(60.000000,8.355450e-07)(65.000000,2.838780e-07)(70.000000,8.890750e-08)(75.000000,2.786550e-08)(80.000000,8.742430e-09)

};

      \addplot[smooth, mark=star,red,only marks] plot coordinates {
(0.000000,7.344340e-01)(5.000000,6.534960e-01)(10.000000,5.448840e-01)(15.000000,4.440380e-01)(20.000000,3.369560e-01)(25.000000,2.312290e-01)(30.000000,1.453700e-01)(35.000000,8.597610e-02)(40.000000,4.893840e-02)(45.000000,2.723670e-02)(50.000000,1.497320e-02)(55.000000,8.183190e-03)(60.000000,4.463870e-03)(65.000000,2.436380e-03)(70.000000,1.332750e-03)(75.000000,7.310500e-04)(80.000000,4.019820e-04)
   };
      \addplot[smooth, mark=star,only marks,red] plot coordinates {
(0.000000,6.497200e-01)(5.000000,4.998840e-01)(10.000000,3.367460e-01)(15.000000,2.110200e-01)(20.000000,1.254530e-01)(25.000000,7.178640e-02)(30.000000,4.006240e-02)(35.000000,2.203020e-02)(40.000000,1.202210e-02)(45.000000,6.541110e-03)(50.000000,3.558940e-03)(55.000000,1.939850e-03)(60.000000,1.060290e-03)(65.000000,5.814290e-04)(70.000000,3.199080e-04)(75.000000,1.765870e-04)(80.000000,9.777600e-05)
};
      \addplot[smooth, mark=star,only marks,red] plot coordinates {
(0.000000,9.957820e-01)(5.000000,3.027220e-01)(10.000000,2.047360e-01)(15.000000,1.242840e-01)(20.000000,7.157770e-02)(25.000000,4.014050e-02)(30.000000,2.217070e-02)(35.000000,1.214060e-02)(40.000000,6.621480e-03)(45.000000,3.608210e-03)(50.000000,1.968480e-03)(55.000000,1.076460e-03)(60.000000,5.904210e-04)(65.000000,3.248710e-04)(70.000000,1.793190e-04)(75.000000,9.926950e-05)(80.000000,5.510000e-05)

};
      \addplot[smooth, mark=star,only marks,red] plot coordinates {
(0.000000,9.943300e-01)(5.000000,2.521880e-01)(10.000000,1.652380e-01)(15.000000,9.770440e-02)(20.000000,5.483490e-02)(25.000000,3.008530e-02)(30.000000,1.635220e-02)(35.000000,8.858170e-03)(40.000000,4.798190e-03)(45.000000,2.603910e-03)(50.000000,1.417360e-03)(55.000000,7.742680e-04)(60.000000,4.245530e-04)(65.000000,2.336490e-04)(70.000000,1.290250e-04)(75.000000,7.146820e-05)(80.000000,3.969490e-05)
};

      \addplot[smooth, mark=star,red,only marks] plot coordinates {
(0.000000,7.441010e-01)(5.000000,5.890720e-01)(10.000000,3.840420e-01)(15.000000,1.942880e-01)(20.000000,7.476300e-02)(25.000000,2.519950e-02)(30.000000,8.095290e-03)(35.000000,2.597250e-03)(40.000000,8.159720e-04)(45.000000,2.562010e-04)(50.000000,8.047700e-05)(55.000000,2.529580e-05)(60.000000,7.956220e-06)(65.000000,2.503910e-06)(70.000000,7.884030e-07)(75.000000,2.483510e-07)(80.000000,7.826070e-08)

   };
      \addplot[smooth, mark=star,only marks,red] plot coordinates {
(0.000000,9.976140e-01)(5.000000,4.401860e-01)(10.000000,1.826500e-01)(15.000000,6.457290e-02)(20.000000,2.108630e-02)(25.000000,6.707870e-03)(30.000000,2.077890e-03)(35.000000,6.567230e-04)(40.000000,2.045790e-04)(45.000000,6.636430e-05)(50.000000,2.079170e-05)(55.000000,6.522400e-06)(60.000000,2.048240e-06)(65.000000,6.437640e-07)(70.000000,2.024780e-07)(75.000000,6.372090e-08)(80.000000,2.006320e-08)

};
      \addplot[smooth, mark=star,only marks,red] plot coordinates {
(0.000000,9.997580e-01)(5.000000,2.733500e-01)(10.000000,1.091070e-01)(15.000000,3.725650e-02)(20.000000,1.203510e-02)(25.000000,3.805430e-03)(30.000000,1.194820e-03)(35.000000,3.732030e-04)(40.000000,1.189560e-04)(45.000000,3.713930e-05)(50.000000,1.169340e-05)(55.000000,3.534850e-06)(60.000000,1.132660e-06)(65.000000,3.666460e-07)(70.000000,1.188100e-07)(75.000000,3.681580e-08)(80.000000,1.173560e-08)

   };
      \addplot[smooth, mark=star,only marks,red] plot coordinates {
(0.000000,9.995120e-01)(5.000000,2.265910e-01)(10.000000,8.572030e-02)(15.000000,2.838630e-02)(20.000000,9.022110e-03)(25.000000,2.823970e-03)(30.000000,8.863630e-04)(35.000000,2.751750e-04)(40.000000,8.563210e-05)(45.000000,2.489640e-05)(50.000000,7.560020e-06)(55.000000,2.897050e-06)(60.000000,8.355450e-07)(65.000000,2.838780e-07)(70.000000,8.890750e-08)(75.000000,2.786550e-08)(80.000000,8.742430e-09)

   };

\addplot[smooth,red,densely dotted, very thick] plot coordinates {
(0.000000,2.506190e+00)(5.000000,2.082550e+00)(10.000000,1.268440e+00)(15.000000,7.151030e-01)(20.000000,3.954910e-01)(25.000000,2.183350e-01)(30.000000,1.208620e-01)(35.000000,6.712460e-02)(40.000000,3.738540e-02)(45.000000,2.086840e-02)(50.000000,1.166860e-02)(55.000000,6.533100e-03)(60.000000,3.661460e-03)(65.000000,2.053640e-03)(70.000000,1.152540e-03)(75.000000,6.471100e-04)(80.000000,3.634590e-04)
};
  \addplot[smooth,red,densely dotted, very thick] plot coordinates {
(0.000000,1.416590e+00)(5.000000,7.069430e-01)(10.000000,3.635370e-01)(15.000000,1.915360e-01)(20.000000,1.028730e-01)(25.000000,5.601860e-02)(30.000000,3.079380e-02)(35.000000,1.703700e-02)(40.000000,9.468040e-03)(45.000000,5.278290e-03)(50.000000,2.949220e-03)(55.000000,1.650580e-03)(60.000000,9.248880e-04)(65.000000,5.187200e-04)(70.000000,2.911160e-04)(75.000000,1.634620e-04)(80.000000,9.181870e-05)
};
\addplot[smooth,red,densely dotted, very thick] plot coordinates {
(0.000000,1.045230e+00)(5.000000,5.009890e-01)(10.000000,2.520270e-01)(15.000000,1.284940e-01)(20.000000,6.664920e-02)(25.000000,3.521660e-02)(30.000000,1.890400e-02)(35.000000,1.027090e-02)(40.000000,5.630440e-03)(45.000000,3.106830e-03)(50.000000,1.722580e-03)(55.000000,9.584920e-04)(60.000000,5.347450e-04)(65.000000,2.989260e-04)(70.000000,1.673490e-04)(75.000000,9.379230e-05)(80.000000,5.261060e-05)

   };
     \addplot[smooth,red,densely dotted,very thick] plot coordinates {
(0.000000,7.747930e-01)(5.000000,3.773640e-01)(10.000000,1.910190e-01)(15.000000,9.754880e-02)(20.000000,5.039110e-02)(25.000000,2.644280e-02)(30.000000,1.409170e-02)(35.000000,7.606930e-03)(40.000000,4.147760e-03)(45.000000,2.278910e-03)(50.000000,1.259320e-03)(55.000000,6.989070e-04)(60.000000,3.891490e-04)(65.000000,2.172090e-04)(70.000000,1.214620e-04)(75.000000,6.801540e-05)(80.000000,3.812670e-05)
 };

\addplot[smooth,red,densely dotted, very thick] plot coordinates {
(0.000000,2.755710e-01)(5.000000,9.461780e-01)(10.000000,4.570880e-01)(15.000000,2.055890e-01)(20.000000,7.663420e-02)(25.000000,2.556240e-02)(30.000000,8.176210e-03)(35.000000,2.617500e-03)(40.000000,8.211810e-04)(45.000000,2.575880e-04)(50.000000,8.085480e-05)(55.000000,2.540060e-05)(60.000000,7.985710e-06)(65.000000,2.512300e-06)(70.000000,7.908170e-07)(75.000000,2.490510e-07)(80.000000,7.846530e-08)

};
  \addplot[smooth,red,densely dotted, very thick] plot coordinates {
(0.000000,1.373760e+00)(5.000000,4.439520e-01)(10.000000,1.712580e-01)(15.000000,6.098710e-02)(20.000000,2.014750e-02)(25.000000,6.464500e-03)(30.000000,2.016010e-03)(35.000000,6.398780e-04)(40.000000,2.000350e-04)(45.000000,6.501570e-05)(50.000000,2.041320e-05)(55.000000,6.414860e-06)(60.000000,2.017370e-06)(65.000000,6.348250e-07)(70.000000,1.998710e-07)(75.000000,6.295600e-08)(80.000000,1.983760e-08)

};
\addplot[smooth,red,densely dotted, very thick] plot coordinates {
(0.000000,1.163470e+00)(5.000000,2.974950e-01)(10.000000,9.607780e-02)(15.000000,3.245810e-02)(20.000000,1.070400e-02)(25.000000,3.452090e-03)(30.000000,1.100440e-03)(35.000000,3.477510e-04)(40.000000,1.116930e-04)(45.000000,3.512270e-05)(50.000000,1.111810e-05)(55.000000,3.382060e-06)(60.000000,1.086790e-06)(65.000000,3.525220e-07)(70.000000,1.144350e-07)(75.000000,3.556580e-08)(80.000000,1.135730e-08)

};
\addplot[smooth,red,densely dotted, very thick] plot coordinates {
(0.000000,8.865710e-01)(5.000000,2.287440e-01)(10.000000,7.217960e-02)(15.000000,2.377170e-02)(20.000000,7.734960e-03)(25.000000,2.477740e-03)(30.000000,7.915620e-04)(35.000000,2.494040e-04)(40.000000,7.850710e-05)(45.000000,2.316030e-05)(50.000000,7.096220e-06)(55.000000,2.699210e-06)(60.000000,7.865370e-07)(65.000000,2.672910e-07)(70.000000,8.406660e-08)(75.000000,2.644530e-08)(80.000000,8.323590e-09)

};

\draw \boundellipse{ axis cs:72.5,3.5e-004}{10}{1.6};
\node at (axis cs:72.5,4e-005){\scriptsize{IM/DD}};

\draw \boundellipse{ axis cs:42.5,1.4e-004}{70}{0.22};
\node at (axis cs:28,1.35e-004){\scriptsize{Heterodyne}};

\legend{BL=16.5 L/min, BL=7.1 L/min, BL=4.7 L/min, BL=2.4 L/min,,,,,Simulation,,,,,,,,High SNR};
      \end{axis}
    \end{tikzpicture}
   \caption{Outage Probability under various levels of air bubbles using salty water for thermally uniform UWOC channels under both IM/DD and heterodyne techniques along with the asymptotic results at high SNR for $\gamma_{\rm{th}}=0$ dB.}
      \label{fig:OP2new}
          \end{center}
\end{figure}

Fig.~\ref{fig:OP2new} presents the outage probability for the dual-hop UWOC system under uniform temperature, various levels of air bubbles, and for both IM/DD and heterodyne techniques using salty water. Expectedly, it can be inferred from Fig.~\ref{fig:OP2new} that for a given type of detection, $P_{\rm{out}}$ increases as the severity of the turbulence increases (i.e. the higher the level of air bubbles, the higher will be the outage probability of the dual-hop UWOC system), leading to a performance deterioration.
In addition, it can also be observed that implementing heterodyne detection results in a significant improvement in the UWOC system performance compared to IM/DD, as expected.
This performance
enhancement is due the fact that heterodyne technique can better
overcome the turbulence effects which comes at the
expense of complexity in implementing coherent receivers relative
to the IM/DD technique \cite{heterodyne1}. For example, for a bubbles level BL=4.7 L/min, to achieve an outage probability of $10^{-3}$, an SNR of 30 dB is required for the heterodyne detection technique while this increases to 55 dB in the case of the IM/DD technique.
\begin{figure}[!h]
   \begin{center}
\begin{tikzpicture}[scale=1.2]
    \begin{axis}[xtick=data,xmin=0,xmax=50,font=\footnotesize,
      ymin=1e-04,ymax=0.5, ymode=log, xlabel=Average Electrical SNR per Hop (dB), ylabel=Average Bit Error Rate,
       legend style={nodes=right},legend style={nodes={scale=0.8, transform shape}},
    legend pos= north east,
    xtick=data,
    xminorgrids,
    grid style={dotted},
    yminorgrids,]
 \addplot[smooth,blue,mark=*,mark options={solid},every mark/.append style={solid, fill=white},thick] plot coordinates {
(0.000000,3.138004e-01)(5.000000,2.053768e-01)(10.000000,1.331135e-01)(15.000000,8.315779e-02)(20.000000,4.953466e-02)(25.000000,2.842693e-02)(30.000000,1.588473e-02)(35.000000,8.744773e-03)(40.000000,4.786489e-03)(45.000000,2.634569e-03)(50.000000,1.456697e-03)
   };
  \addplot[smooth,blue,mark=*,mark options={solid},every mark/.append style={solid, fill=white},thick] plot coordinates {
(0.000000,2.944055e-01)(5.000000,1.548841e-01)(10.000000,7.617947e-02)(15.000000,4.503048e-02)(20.000000,2.566385e-02)(25.000000,1.426879e-02)(30.000000,7.805713e-03)(35.000000,4.241906e-03)(40.000000,2.295548e-03)(45.000000,1.242179e-03)(50.000000,6.741055e-04)
   };
  \addplot[smooth,blue,mark=*,mark options={solid},every mark/.append style={solid, fill=white},thick] plot coordinates {
(0.000000,2.818800e-01)(5.000000,1.213835e-01)(10.000000,1.462904e-02)(15.000000,4.215823e-05)(20.000000,1.539240e-12)(25.000000,5.323861e-34)(30.000000,1.243342e-93)(35.000000,8.962400e-279)(40.000000,0.000000e+00)(45.000000,0.000000e+00)(50.000000,0.000000e+00)
   };

      \addplot[smooth, mark=star,only marks,red] plot coordinates {
(0.000000,3.138004e-01)(5.000000,2.053768e-01)(10.000000,1.331135e-01)(15.000000,8.315779e-02)(20.000000,4.953466e-02)(25.000000,2.842693e-02)(30.000000,1.588473e-02)(35.000000,8.744773e-03)(40.000000,4.786489e-03)(45.000000,2.634569e-03)(50.000000,1.456697e-03)
};
      \addplot[smooth, mark=star,only marks,red] plot coordinates {
(0.000000,2.944055e-01)(5.000000,1.548841e-01)(10.000000,7.617947e-02)(15.000000,4.503048e-02)(20.000000,2.566385e-02)(25.000000,1.426879e-02)(30.000000,7.805713e-03)(35.000000,4.241906e-03)(40.000000,2.295548e-03)(45.000000,1.242179e-03)(50.000000,6.741055e-04)
   };
      \addplot[smooth, mark=star,only marks,red] plot coordinates {
(0.000000,2.818800e-01)(5.000000,1.213835e-01)(10.000000,1.462904e-02)(15.000000,4.215823e-05)(20.000000,1.539240e-12)(25.000000,5.323861e-34)(30.000000,1.243342e-93)(35.000000,8.962400e-279)(40.000000,0.000000e+00)(45.000000,0.000000e+00)(50.000000,0.000000e+00)
   };
  \addplot[smooth,tension=0.01,red, densely dotted ,every mark/.append style={solid, fill=white},very thick] plot coordinates {
(0.000000,4.539220e+14)(5.000000,3.455100e+05)(10.000000,1.730490e-01)(15.000000,8.906070e-02)(20.000000,4.766990e-02)(25.000000,2.589720e-02)(30.000000,1.421340e-02)(35.000000,7.855300e-03)(40.000000,4.362230e-03)(45.000000,2.430610e-03)(50.000000,1.357590e-03)
};
  \addplot[smooth,tension=0.01,red, densely dotted, every mark/.append style={solid, fill=white},very thick] plot coordinates {
(0.000000,2.068310e+37)(5.000000,1.394720e+18)(10.000000,1.242150e+05)(15.000000,4.855650e-02)(20.000000,2.515190e-02)(25.000000,1.325610e-02)(30.000000,7.096200e-03)(35.000000,3.845960e-03)(40.000000,2.103970e-03)(45.000000,1.159030e-03)(50.000000,6.417910e-04)
   };

      \addplot[smooth,densely dashed,green!70!black,thick] plot coordinates {
(0.000000,1.928699e-01)(5.000000,1.180514e-01)(10.000000,8.067149e-02)(15.000000,5.252324e-02)(20.000000,3.231265e-02)(25.000000,1.917041e-02)(30.000000,1.112362e-02)(35.000000,6.354998e-03)(40.000000,3.595642e-03)(45.000000,2.019004e-03)(50.000000,1.127615e-03)
};
    \addplot[smooth,densely dashed,green!70!black,thick] plot coordinates {
(0.000000,1.718820e-01)(5.000000,7.278245e-02)(10.000000,3.988382e-02)(15.000000,2.598836e-02)(20.000000,1.603425e-02)(25.000000,9.522190e-03)(30.000000,5.522454e-03)(35.000000,3.154677e-03)(40.000000,1.781856e-03)(45.000000,1.002054e-03)(50.000000,5.664769e-04)
   };
    \addplot[smooth,densely dashed,green!70!black,thick] plot coordinates {
(0.000000,1.586989e-01)(5.000000,3.776231e-02)(10.000000,7.982275e-04)(15.000000,1.129764e-08)(20.000000,5.091466e-23)(25.000000,1.086020e-63)(30.000000,1.910259e-185)(35.000000,0.000000e+00)(40.000000,0.000000e+00)(45.000000,0.000000e+00)(50.000000,0.000000e+00)
   };

\draw \boundellipse{ axis cs:37.5,6e-003}{8}{0.4};
\node at (axis cs:43,1.15e-002){\scriptsize{BL$=$7.1L/min}};

\draw \boundellipse{ axis cs:42.5,1.5e-003}{8}{0.4};
\node at (axis cs:35,1.2e-003){\scriptsize{BL$=$4.7L/min}};

\draw \boundellipse{ axis cs:11,1.2e-003}{28}{0.13};
\node at (axis cs:19.5,1.2e-003){\scriptsize{BL$=$0L/min}};
\legend{Dual-Hop Analytic,,,Dual-Hop Simulation,,,Dual-Hop Asymptotic,,,Single Link Simulation};
      \end{axis}
    \end{tikzpicture}
   \caption{Average BER for OOK of single UWOC and dual-hop UWOC links under various levels of air bubbles using fresh water for thermally uniform UWOC channels operating using IM/DD technique along with the asymptotic results at high SNR.}
      \label{fig:BER1new}
          \end{center}
\end{figure}
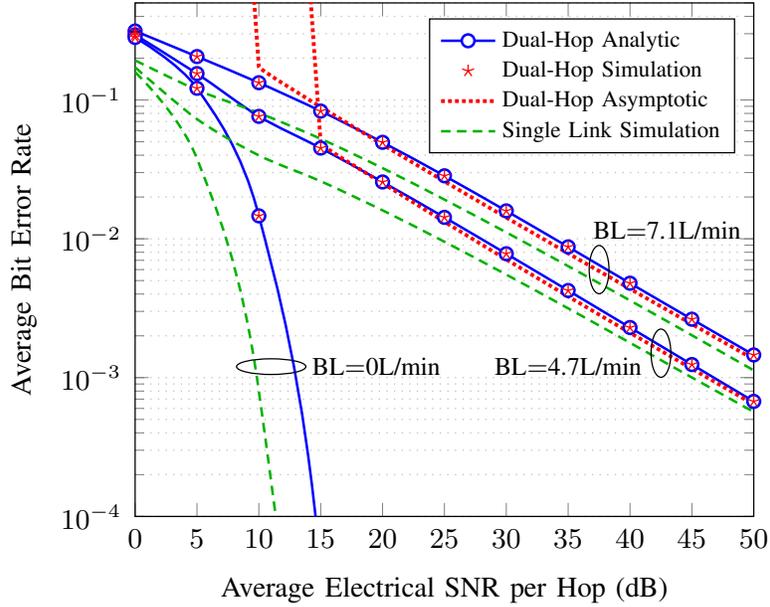

The average BER performance of the dual-hop system with the end-to-end link length of 2 m in operation under the IM/DD technique (i.e. $r_1=2$ and $r_2=2$) for different levels of air bubbles is illustrated in Fig.~\ref{fig:BER1new} in the case
of uniform temperature and fresh water.
The asymptotic results of the average BER of the dual-hop UWOC system at high SNR values obtained by utilizing (\ref{BERHighSNR}) are also shown in Fig.~\ref{fig:BER1new}. Simulation results for a 1 m single UWOC link under the same channel conditions are also included for comparison purposes.
We can see from this figure that the analytical results of the average BER are in a good match with the Monte-Carlo simulated results. Moreover, it can be observed that for both dual-hop and single UWOC links, the average BER performance degrades as the level of air bubbles increases, resulting in a severe turbulence condition.
One of the most important outcomes of Fig.~\ref{fig:BER1new} is that
the dual-hop UWOC system,
 where each hop has the length of 1 m, offers less BER performance for all channel conditions,
 as compared with the single UWOC link with a total length of 1 m. This is due to the fact that the degrading effects of scattering, absorption and turbulence-induced fading
increase with the distance. Hence, this result emphasizes the effectiveness
  of the dual-hop system
in mitigating the short range problem in UWOC with low power requirements and the impairment effects in the underwater medium.
Furthermore, it can be clearly seen from Fig.~\ref{fig:BER1new} that the asymptotic result of the average BER given by (\ref{BERHighSNR}) matches perfectly the analytical closed-form result in (\ref{SNRBER}) at high SNR regime, proving the accuracy of our asymptotic analysis.
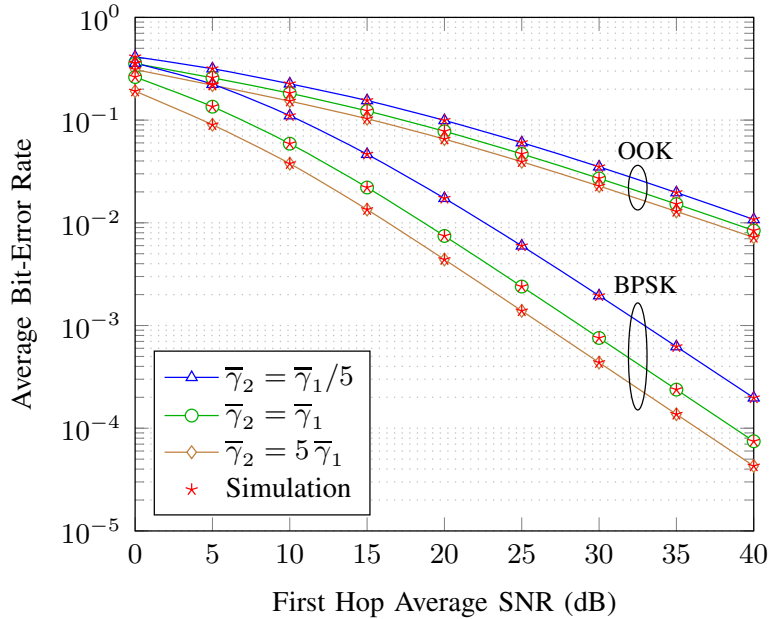
\begin{figure}[!h]
  \begin{center}
\begin{tikzpicture}[scale=1.2]
    \begin{axis}[xtick=data,font=\footnotesize,
    ymode=log, xlabel= First Hop Average SNR (dB), ylabel= Average Bit-Error Rate,
  xmin=0,xmax=40,ymin=1e-05, ymax=1,
    legend style={nodes=right},legend pos= south west,
    xtick=data,
      xminorgrids,
    grid style={dotted},
    yminorgrids,
   ]
     \addplot[smooth,blue,mark=triangle*,mark options={solid},every mark/.append style={solid, fill=white}] plot coordinates {
(0.000000,4.129136e-01)(5.000000,3.165488e-01)(10.000000,2.255993e-01)(15.000000,1.550521e-01)(20.000000,9.932274e-02)(25.000000,6.036590e-02)(30.000000,3.508231e-02)(35.000000,1.969392e-02)(40.000000,1.078043e-02)(45.000000,5.812003e-03)(50.000000,3.111770e-03)
};
    \addplot[smooth,
green!70!black,mark=*,mark options={solid},every mark/.append style={solid, fill=white}] plot coordinates {
(0.000000,3.552589e-01)(5.000000,2.580412e-01)(10.000000,1.826408e-01)(15.000000,1.228161e-01)(20.000000,7.785611e-02)(25.000000,4.683704e-02)(30.000000,2.707286e-02)(35.000000,1.522483e-02)(40.000000,8.407446e-03)(45.000000,4.604491e-03)(50.000000,2.508514e-03)
};
    \addplot[smooth,brown,mark=diamond*,mark options={solid},every mark/.append style={solid, fill=white}] plot coordinates {
(0.000000,3.104644e-01)(5.000000,2.188530e-01)(10.000000,1.530180e-01)(15.000000,1.030430e-01)(20.000000,6.531077e-02)(25.000000,3.930062e-02)(30.000000,2.280067e-02)(35.000000,1.290663e-02)(40.000000,7.214023e-03)(45.000000,4.014663e-03)(50.000000,2.231988e-03)
};
     \addplot[smooth,blue,mark=triangle*,mark options={solid},every mark/.append style={solid, fill=white}] plot coordinates {
(0.000000,3.587177e-01)(5.000000,2.233306e-01)(10.000000,1.104814e-01)(15.000000,4.648950e-02)(20.000000,1.732339e-02)(25.000000,5.966500e-03)(30.000000,1.953470e-03)(35.000000,6.230747e-04)(40.000000,1.978188e-04)(45.000000,6.328391e-05)(50.000000,2.179716e-05)
};
    \addplot[smooth,green!70!black,mark=*,mark options={solid},every mark/.append style={solid, fill=white}] plot coordinates {
(0.000000,2.626830e-01)(5.000000,1.351020e-01)(10.000000,5.892620e-02)(15.000000,2.207820e-02)(20.000000,7.450820e-03)(25.000000,2.395300e-03)(30.000000,7.566370e-04)(35.000000,2.377680e-04)(40.000000,7.464430e-05)(45.000000,2.344220e-05)(50.000000,1.269150e+01)
};
    \addplot[smooth,brown,mark=diamond*,mark options={solid},every mark/.append style={solid, fill=white}] plot coordinates {
(0.000000,1.915370e-01)(5.000000,9.008770e-02)(10.000000,3.769550e-02)(15.000000,1.345750e-02)(20.000000,4.393030e-03)(25.000000,1.389620e-03)(30.000000,4.360270e-04)(35.000000,1.367010e-04)(40.000000,4.289790e-05)(45.000000,1.348060e-05)(50.000000,1.269180e+01)
};
\addplot[smooth, mark=star,only marks,red] plot coordinates {
(0.000000,4.129136e-01)(5.000000,3.165488e-01)(10.000000,2.255993e-01)(15.000000,1.550521e-01)(20.000000,9.932274e-02)(25.000000,6.036590e-02)(30.000000,3.508231e-02)(35.000000,1.969392e-02)(40.000000,1.078043e-02)(45.000000,5.812003e-03)(50.000000,3.111770e-03)
};
\addplot[smooth, mark=star,only marks,red] plot coordinates {
(0.000000,3.552589e-01)(5.000000,2.580412e-01)(10.000000,1.826408e-01)(15.000000,1.228161e-01)(20.000000,7.785611e-02)(25.000000,4.683704e-02)(30.000000,2.707286e-02)(35.000000,1.522483e-02)(40.000000,8.407446e-03)(45.000000,4.604491e-03)(50.000000,2.508514e-03)
};
\addplot[smooth, mark=star,only marks,red] plot coordinates {
(0.000000,3.104644e-01)(5.000000,2.188530e-01)(10.000000,1.530180e-01)(15.000000,1.030430e-01)(20.000000,6.531077e-02)(25.000000,3.930062e-02)(30.000000,2.280067e-02)(35.000000,1.290663e-02)(40.000000,7.214023e-03)(45.000000,4.014663e-03)(50.000000,2.231988e-03)
};
\addplot[smooth, mark=star,only marks,red] plot coordinates {
(0.000000,3.587177e-01)(5.000000,2.233306e-01)(10.000000,1.104814e-01)(15.000000,4.648950e-02)(20.000000,1.732339e-02)(25.000000,5.966500e-03)(30.000000,1.953470e-03)(35.000000,6.230747e-04)(40.000000,1.978188e-04)(45.000000,6.328391e-05)(50.000000,2.179716e-05)
};
\addplot[smooth, mark=star,only marks,red] plot coordinates {
(0.000000,2.626830e-01)(5.000000,1.351020e-01)(10.000000,5.892620e-02)(15.000000,2.207820e-02)(20.000000,7.450820e-03)(25.000000,2.395300e-03)(30.000000,7.566370e-04)(35.000000,2.377680e-04)(40.000000,7.464430e-05)(45.000000,2.344220e-05)(50.000000,1.269150e+01)
};
\addplot[smooth, mark=star,only marks,red] plot coordinates {
(0.000000,1.915370e-01)(5.000000,9.008770e-02)(10.000000,3.769550e-02)(15.000000,1.345750e-02)(20.000000,4.393030e-03)(25.000000,1.389620e-03)(30.000000,4.360270e-04)(35.000000,1.367010e-04)(40.000000,4.289790e-05)(45.000000,1.348060e-05)(50.000000,1.269180e+01)
};

\legend{$\overline{\gamma}_2=\overline{\gamma}_1/5$, $\overline{\gamma}_2=\overline{\gamma}_1$,$\overline{\gamma}_2=5 \,\overline{\gamma}_1$,,,,,Simulation};

 \draw \boundellipse{ axis cs:32.5,2.2e-002}{6}{0.5};
\node at (axis cs:33,5e-002){\scriptsize{OOK}};

\draw \boundellipse{ axis cs:32.5,5e-004}{6}{1.2};
\node at (axis cs:33,2.5e-003){\scriptsize{BPSK}};

 \end{axis}
  \end{tikzpicture}
     \caption{Average BER for OOK and BPSK of dual-hop UWOC systems with balanced or unbalanced hops for a bubbles level of 4.7 L/min and a temperature gradient of 0.05 $^\circ$C.$cm^{-1}$.}
          \label{fig:BERnew}
     \end{center}
  \end{figure}

Fig.~\ref{fig:BERnew} illustrates the effect of the power imbalance between the two hops on the BER performance under both IM/DD with OOK and heterodyne detection with BPSK modulation schemes.
We can observe that the imbalance between the hops can be either beneficial or deleterious for the overall system performance. More specifically, it can be seen from Fig.~\ref{fig:BERnew} that there is an improvement in the average BER performance when $\overline{\gamma}_2 > \overline{\gamma}_1$ which is greater in the case of heterodyne technique, and there is a degradation otherwise. In addition, it can be shown that BPSK modulation always performs better than OOK for all SNR range, as expected.

Fig.~\ref{fig:BER2new} depicts the average BER performance of the dual-hop UWOC system operating under the heterodyne detection technique (i.e. $r_1=1$ and $r_2=1$) for 64-QAM, 16-PSK, 16 QAM, and BPSK modulation schemes in the case of strong turbulence conditions corresponding to a level of bubbles equal to 4.7 L/min and a temperature gradient of 0.10 $^\circ$C.$cm^{-1}$, with a scintillation index $\sigma_I^2=0.4769$. Clearly, it can be observed from Fig.~\ref{fig:BER2new} that BPSK modulation offers the best performance compared to the presented modulation techniques.
In addition, it can be seen from Fig.~\ref{fig:BER2new} that 16-QAM outperforms 16-PSK, as expected when $M > 4$ \cite{proakis2008digital}. Finally, it can also be noticed that the derived asymptotic results at high SNR range are very tight and accurate.
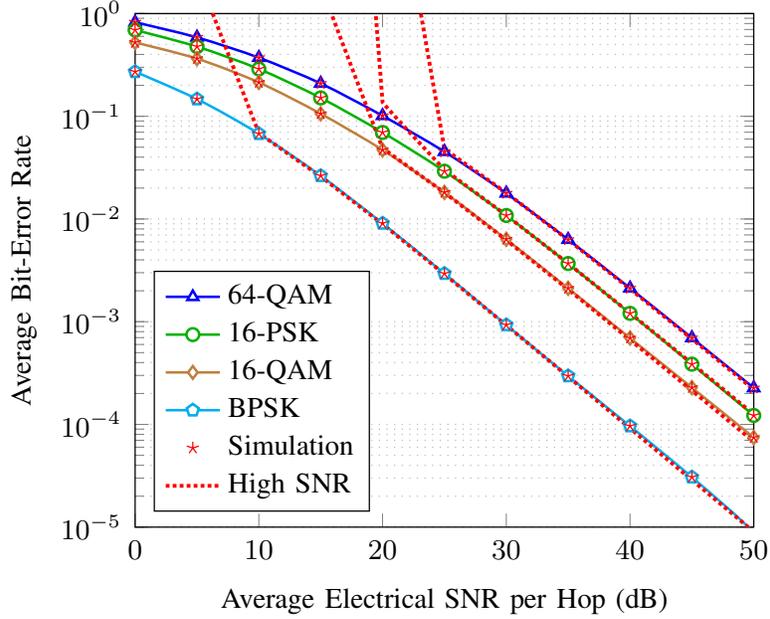
\begin{figure}[!h]
  \begin{center}
\begin{tikzpicture}[scale=1.2]
    \begin{axis}[font=\footnotesize,
    ymode=log, xlabel= Average Electrical SNR per Hop (dB), ylabel= Average Bit-Error Rate,
  xmin=0,xmax=50,ymin=1e-05, ymax=1,
    legend style={nodes=right},legend pos= south west,
      xminorgrids,
    grid style={dotted},
    yminorgrids,
   ]

     \addplot[smooth,blue,mark=triangle*,mark options={solid},every mark/.append style={solid, fill=white},thick] plot coordinates {
(0.000000,8.236960e-01)(5.000000,5.863756e-01)(10.000000,3.721336e-01)(15.000000,2.086872e-01)(20.000000,1.011961e-01)(25.000000,4.518831e-02)(30.000000,1.779617e-02)(35.000000,6.303914e-03)(40.000000,2.117778e-03)(45.000000,6.955396e-04)(50.000000,2.263889e-04)(55.000000,7.250874e-05)(60.000000,2.303945e-05)
};
    \addplot[smooth,
green!70!black,mark=*,mark options={solid},every mark/.append style={solid, fill=white},thick] plot coordinates {
(0.000000,6.921009e-01)(5.000000,4.766898e-01)(10.000000,2.887950e-01)(15.000000,1.509619e-01)(20.000000,6.958607e-02)(25.000000,2.919501e-02)(30.000000,1.081610e-02)(35.000000,3.685825e-03)(40.000000,1.205673e-03)(45.000000,3.860925e-04)(50.000000,1.221989e-04)(55.000000,3.927081e-05)(60.000000,1.262337e-05)
};
    \addplot[smooth,brown,mark=diamond*,mark options={solid},every mark/.append style={solid, fill=white},thick] plot coordinates {
(0.000000,5.230832e-01)(5.000000,3.641377e-01)(10.000000,2.137156e-01)(15.000000,1.058051e-01)(20.000000,4.699502e-02)(25.000000,1.814783e-02)(30.000000,6.341729e-03)(35.000000,2.114282e-03)(40.000000,6.935538e-04)(45.000000,2.282153e-04)(50.000000,7.448512e-05)(55.000000,2.376221e-05)(60.000000,7.207434e-06)
};
    \addplot[smooth,cyan,mark=pentagon*,mark options={solid},every mark/.append style={solid, fill=white},thick] plot coordinates {
(0.000000,2.716132e-01)(5.000000,1.471391e-01)(10.000000,6.768528e-02)(15.000000,2.636655e-02)(20.000000,9.072619e-03)(25.000000,2.941158e-03)(30.000000,9.300510e-04)(35.000000,2.963021e-04)(40.000000,9.638922e-05)(45.000000,3.066936e-05)(50.000000,8.857550e-06)(55.000000,1.721098e-06)(60.000000,1.137702e-07)
};
\addplot[smooth, mark=star,only marks,red] plot coordinates {
(0.000000,8.236960e-01)(5.000000,5.863756e-01)(10.000000,3.721336e-01)(15.000000,2.086872e-01)(20.000000,1.011961e-01)(25.000000,4.518831e-02)(30.000000,1.779617e-02)(35.000000,6.303914e-03)(40.000000,2.117778e-03)(45.000000,6.955396e-04)(50.000000,2.263889e-04)(55.000000,7.250874e-05)(60.000000,2.303945e-05)
};
   \addplot[smooth, mark=star,only marks,red] plot coordinates {
(0.000000,6.921009e-01)(5.000000,4.766898e-01)(10.000000,2.887950e-01)(15.000000,1.509619e-01)(20.000000,6.958607e-02)(25.000000,2.919501e-02)(30.000000,1.081610e-02)(35.000000,3.685825e-03)(40.000000,1.205673e-03)(45.000000,3.860925e-04)(50.000000,1.221989e-04)(55.000000,3.927081e-05)(60.000000,1.262337e-05)
};
   \addplot[smooth, mark=star,only marks,red] plot coordinates {
(0.000000,5.230832e-01)(5.000000,3.641377e-01)(10.000000,2.137156e-01)(15.000000,1.058051e-01)(20.000000,4.699502e-02)(25.000000,1.814783e-02)(30.000000,6.341729e-03)(35.000000,2.114282e-03)(40.000000,6.935538e-04)(45.000000,2.282153e-04)(50.000000,7.448512e-05)(55.000000,2.376221e-05)(60.000000,7.207434e-06)
};
   \addplot[smooth, mark=star,only marks,red] plot coordinates {
(0.000000,2.716132e-01)(5.000000,1.471391e-01)(10.000000,6.768528e-02)(15.000000,2.636655e-02)(20.000000,9.072619e-03)(25.000000,2.941158e-03)(30.000000,9.300510e-04)(35.000000,2.963021e-04)(40.000000,9.638922e-05)(45.000000,3.066936e-05)(50.000000,8.857550e-06)(55.000000,1.721098e-06)(60.000000,1.137702e-07)
};
     \addplot[smooth,tension=0.01,red,densely dotted, very thick] plot coordinates {
(0.000000,3.933640e+35)(5.000000,3.796580e+26)(10.000000,2.731360e+18)(15.000000,1.965140e+10)(20.000000,1.415110e+02)(25.000000,4.858410e-02)(30.000000,1.799400e-02)(35.000000,6.265350e-03)(40.000000,2.083860e-03)(45.000000,6.800970e-04)(50.000000,2.185650e-04)(55.000000,7.018020e-05)(60.000000,2.267550e-05)
};
    \addplot[smooth,tension=0.01,red,densely dotted,very thick] plot coordinates {
(0.000000,1.622020e+32)(5.000000,1.565510e+23)(10.000000,1.126270e+15)(15.000000,8.103170e+06)(20.000000,1.376140e-01)(25.000000,3.022110e-02)(30.000000,1.085210e-02)(35.000000,3.692440e-03)(40.000000,1.211370e-03)(45.000000,3.922600e-04)(50.000000,1.254080e-04)(55.000000,4.012780e-05)(60.000000,1.293500e-05)
};
   \addplot[smooth,tension=0.01,red,densely dotted,very  thick] plot coordinates {
(0.000000,3.578640e+25)(5.000000,3.453940e+16)(10.000000,2.484860e+08)(15.000000,1.912870e+00)(20.000000,4.898310e-02)(25.000000,1.807670e-02)(30.000000,6.229990e-03)(35.000000,2.063320e-03)(40.000000,6.662530e-04)(45.000000,2.137330e-04)(50.000000,6.789120e-05)(55.000000,2.162710e-05)(60.000000,6.950050e-06)
};
 \addplot[smooth,tension=0.01,red,densely dotted,very thick] plot coordinates {
(0.000000,2.469920e+09)(5.000000,2.559520e+00)(10.000000,6.855910e-02)(15.000000,2.591020e-02)(20.000000,8.912560e-03)(25.000000,2.899350e-03)(30.000000,9.226560e-04)(35.000000,2.918170e-04)(40.000000,9.150890e-05)(45.000000,2.879580e-05)(50.000000,9.013990e-06)(55.000000,2.841040e-06)(60.000000,9.061330e-07)
};
\legend{64-QAM, 16-PSK,16-QAM,BPSK,Simulation,,,,,High SNR};
%
%
\end{axis}
  \end{tikzpicture}
     \caption{Average BER for different modulation schemes of dual-hop UWOC systems operating under heterodyne detection along with the asymptotic results at high SNR for a bubbles level of 4.7 L/min and a temperature gradient of 0.10 $^\circ$C.$cm^{-1}$.}
          \label{fig:BER2new}
     \end{center}
  \end{figure}
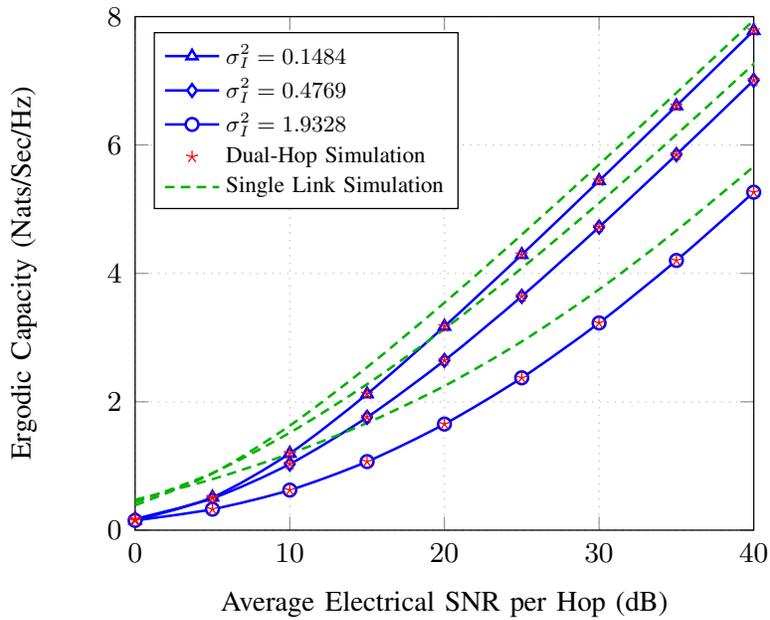
\begin{figure}[!h]
   \begin{center}
\begin{tikzpicture}[scale=1.2]
    \begin{axis}[font=\footnotesize,xmin=0,xmax=40,
   ymin=0,ymax=8, xlabel=Average Electrical SNR per Hop (dB), ylabel= Ergodic Capacity (Nats/Sec/Hz),
       legend style={nodes=right},
    legend pos= north west,legend style={nodes={scale=0.8, transform shape}},
    xmajorgrids,
    grid style={dotted},
    ymajorgrids,]

    \addplot[smooth,blue,mark=triangle*,mark options={solid},every mark/.append style={solid, fill=white},thick] plot coordinates {
(0.000000,1.531799e-01)(5.000000,5.119999e-01)(10.000000,1.195531e+00)(15.000000,2.120716e+00)(20.000000,3.173824e+00)(25.000000,4.291602e+00)(30.000000,5.442328e+00)(35.000000,6.608918e+00)(40.000000,7.781832e+00)(45.000000,8.955869e+00)(50.000000,1.012827e+01)(55.000000,1.129771e+01)(60.000000,1.246396e+01)
   };
%


    \addplot[smooth,blue,mark=diamond*,mark options={solid},every mark/.append style={solid, fill=white},thick] plot coordinates {
(0.000000,1.756599e-01)(5.000000,4.966399e-01)(10.000000,1.030228e+00)(15.000000,1.758941e+00)(20.000000,2.643776e+00)(25.000000,3.642970e+00)(30.000000,4.721364e+00)(35.000000,5.850350e+00)(40.000000,7.008091e+00)(45.000000,8.179668e+00)(50.000000,9.355917e+00)(55.000000,1.053181e+01)(60.000000,1.170509e+01)
   };
    \addplot[smooth,blue,mark=*,mark options={solid},every mark/.append style={solid, fill=white},thick] plot coordinates {
(0.000000,1.518203e-01)(5.000000,3.270330e-01)(10.000000,6.246327e-01)(15.000000,1.068515e+00)(20.000000,1.652829e+00)(25.000000,2.373422e+00)(30.000000,3.228158e+00)(35.000000,4.201558e+00)(40.000000,5.265086e+00)(45.000000,6.387735e+00)(50.000000,7.544012e+00)(55.000000,8.716334e+00)(60.000000,9.894056e+00)
   };

 \addplot[smooth,color=red,solid, mark=star,only marks] plot coordinates {
(0.000000,1.531799e-01)(5.000000,5.119999e-01)(10.000000,1.195531e+00)(15.000000,2.120716e+00)(20.000000,3.173824e+00)(25.000000,4.291602e+00)(30.000000,5.442328e+00)(35.000000,6.608918e+00)(40.000000,7.781832e+00)(45.000000,8.955869e+00)(50.000000,1.012827e+01)(55.000000,1.129771e+01)(60.000000,1.246396e+01)
   };
 \addplot[smooth,color=red,solid, mark=star,only marks] plot coordinates {
(0.000000,1.756599e-01)(5.000000,4.966399e-01)(10.000000,1.030228e+00)(15.000000,1.758941e+00)(20.000000,2.643776e+00)(25.000000,3.642970e+00)(30.000000,4.721364e+00)(35.000000,5.850350e+00)(40.000000,7.008091e+00)(45.000000,8.179668e+00)(50.000000,9.355917e+00)(55.000000,1.053181e+01)(60.000000,1.170509e+01)
   };
 \addplot[smooth,color=red,solid, mark=star,only marks] plot coordinates {
(0.000000,1.518203e-01)(5.000000,3.270330e-01)(10.000000,6.246327e-01)(15.000000,1.068515e+00)(20.000000,1.652829e+00)(25.000000,2.373422e+00)(30.000000,3.228158e+00)(35.000000,4.201558e+00)(40.000000,5.265086e+00)(45.000000,6.387735e+00)(50.000000,7.544012e+00)(55.000000,8.716334e+00)(60.000000,9.894056e+00)
   };

       \addplot[smooth,densely dashed,
green!70!black,thick] plot coordinates {
(0.000000,3.886667e-01)(5.000000,8.822131e-01)(10.000000,1.627757e+00)(15.000000,2.539584e+00)(20.000000,3.544107e+00)(25.000000,4.603158e+00)(30.000000,5.696860e+00)(35.000000,6.812928e+00)(40.000000,7.943126e+00)(45.000000,9.082002e+00)(50.000000,1.022608e+01)(55.000000,1.137323e+01)(60.000000,1.252216e+01)
   };
      \addplot[smooth,densely dashed,
green!70!black,thick] plot coordinates {
(0.000000,4.338734e-01)(5.000000,8.901744e-01)(10.000000,1.518050e+00)(15.000000,2.274333e+00)(20.000000,3.134510e+00)(25.000000,4.081647e+00)(30.000000,5.097350e+00)(35.000000,6.162739e+00)(40.000000,7.261537e+00)(45.000000,8.381570e+00)(50.000000,9.514590e+00)(55.000000,1.065534e+01)(60.000000,1.180062e+01)
   };
      \addplot[smooth,densely dashed,
green!70!black,thick] plot coordinates {
(0.000000,4.733663e-01)(5.000000,7.941951e-01)(10.000000,1.192001e+00)(15.000000,1.670077e+00)(20.000000,2.246424e+00)(25.000000,2.938831e+00)(30.000000,3.750362e+00)(35.000000,4.667055e+00)(40.000000,5.665504e+00)(45.000000,6.721608e+00)(50.000000,7.815553e+00)(55.000000,8.933081e+00)(60.000000,1.006478e+01)
   };

\legend{$\sigma_I^2=0.1484$, $\sigma_I^2=0.4769$,$\sigma_I^2=1.9328$,Dual-Hop Simulation,,,Single Link Simulation};
      \end{axis}
    \end{tikzpicture}
    \caption{Ergodic capacity of dual-hop UWOC systems using IM/DD for different levels of air bubbles and temperature gradients.}
      \label{fig:capacity1new}
          \end{center}
    \end{figure}

The ergodic capacity of the dual-hop 2 m long UWOC system in operation under the IM/DD technique is presented in Fig.~\ref{fig:capacity1new} for different levels of air bubbles and temperature gradients, associated with different scintillation index values along with the Monte-Carlo simulation capacity results of the single 1 m long UWOC link. It can be seen from this figure that as the effect of the air bubbles and/or the gradient of temperature increases, the scintillation index increases and therefore the ergodic capacity degrades, for both dual-hop as well as single UWOC links. Similar to Fig.~\ref{fig:BER1new}, it can be noticed from this figure that the turbulence-induced fading is an incremental function of the distance.
This observation justifies the advantage of dividing the long communication distance to shorter ones by means of intermediate relays in mitigating the short range issue and the turbulence-induced fading in UWOC.
\section{Conclusion}
In this paper, we have studied the performance of a dual-hop UWOC system
using AF fixed-gain relaying operating under both IM/DD and heterodyne detection over mixture EGG fading channels in the presence of both temperature gradients as well as air bubbles induced turbulence, for both fresh as well as salty waters. The EGG model has been shown to provide an excellent fit to the measured data acquired from an indoor laboratory experiment and has
a simple mathematical form, making it attractive from a performance analysis point of view.
We have derived closed-form expressions
for the PDF and the CDF of the end-to-end SNR in terms
of the bivariate Fox's H function. Moreover, based on these formulas, we have obtained exact closed-form expressions for fundamental system performance metrics such as the outage probability, the average BER for different modulation schemes, and the ergodic capacity under different turbulence conditions for both fresh and salty waters. Moreover, we have presented very tight asymptotic results for the obtained performance metrics at high SNR in terms of simple elementary functions.
We have also demonstrated the capability of dual-hop UWOC systems in mitigating the short range as well as the turbulence-induced fading issues. Finally, we have shown that the presence of air bubbles and gradient of temperature can severely degrade the end-to-end performance of both single link as well as dual-hop UWOC links.

\appendices
\renewcommand{\theequation}{\thesection.\arabic{equation}}
\setcounter{equation}{0}
\section{Proof of Theorem 1}\label{A}
In this appendix, we present the CDF of the overall SNR, $\gamma$, for a dual-hop UWOC system using AF fixed gain relaying
\begin{align}\label{CDFUWOCP1}
F_\gamma(\gamma)={\rm{Pr}}\left [ \frac{\gamma_1 \gamma_2}{\gamma_2+C}  \leq   \gamma \right],
\end{align}
where ${\rm{Pr}}[A]$ represents the probability of an event $A$. Then the CDF can be written as
\begin{align}\label{CDFUWOCP2}
\nonumber  F_\gamma(\gamma)&=\int_{0}^{\infty}{\rm{Pr}}\left [ \frac{\gamma_1 \gamma_2}{\gamma_2+C}  \leq  \gamma|\gamma_2 \right]f_{\gamma_2}(\gamma_2)\,d\gamma_2\\
&=1-\int_{\gamma}^{\infty}\overline{F}_{\gamma_2}\left ( \frac{C \gamma}{x-\gamma} \right )f_{\gamma_1}(x)\,dx,
\end{align}
where $\overline{F}_{\gamma_2}$ is the complementary CDF of $\gamma_2$ that can be expressed using \cite[Eqs.(07.34.03.0283.01), (07.34.03.0275.01), (06.07.03.0002.01), and (07.34.03.0613.01)]{Wolfram} as

\begin{align}\label{CCDF}
 \overline{F}_{\gamma_2}(\gamma_2)=\omega_2 \exp\left ( -\frac{1}{\lambda_2} \left ( \frac{\gamma_2}{\mu_{r_{2}}} \right )^{\frac{1}{r_2}}\right )
+\frac{(1-\omega_2)}{\Gamma(a_2)}{\rm{G}}_{1,2}^{2,0}\left[ \frac{1}{b_2^{c_2}}\left ( \frac{\gamma_2}{\mu_{r_{2}}} \right )^{\frac{c_2}{r_2}}\left| \begin{matrix} {1} \\ {a_2,0} \\ \end{matrix} \right.\right].
\end{align}\noindent
Substituting (\ref{SNRPDFHop}) and (\ref{CCDF}) in (\ref{CDFUWOCP2}), then using the change of variable $z=x-\gamma$ we obtain
\begin{align}\label{CDFUWOCP3}
F_\gamma(\gamma)=1-\mathcal{I}_1-\mathcal{I}_2-\mathcal{I}_3-\mathcal{I}_4,
\end{align}
with
\begin{align}\label{I1P1}
\mathcal{I}_1=\frac{\omega_1 \omega_2}{r_1}\int_{0}^{\infty}\frac{1}{(z+\gamma)}
\exp\left ( -\frac{1}{\lambda_2} \left ( \frac{C\gamma }{\mu_{r_{2}}z} \right )^{\frac{1}{r_2}}\right )
{\rm{G}}_{0,1}^{1,0}\left[ \frac{1}{\lambda_1}\left (\frac{ z+\gamma}{\mu_{r_1}} \right )^{\frac{1}{r_1}}\left| \begin{matrix} {-} \\ {1} \\ \end{matrix} \right.\right]dz.
\end{align}

\begin{align}\label{I2P1}
 \mathcal{I}_2=\frac{\omega_1 \left (1-\omega_2  \right )}{r_1 \Gamma(a_2)}\int_{0}^{\infty}\frac{1}{(z+\gamma)}
{\rm{G}}_{0,1}^{1,0}\left[ \frac{1}{\lambda_1}\left (\frac{ z+\gamma}{\mu_{r_1}} \right )^{\frac{1}{r_1}}\left| \begin{matrix} {-} \\ {1} \\ \end{matrix} \right.\right]
{\rm{G}}_{1,2}^{2,0}\left[ \frac{1}{b_2^{c_2}}\left (\frac{ C \gamma}{\mu_{r_2}z} \right )^{\frac{c_2}{r_2}}\left| \begin{matrix} {1} \\ {a_2,0} \\ \end{matrix} \right.\right]dz.                                                                                                                   \end{align}

\begin{align}\label{I3P1}
\mathcal{I}_3=\frac{c_1\omega_2 \left (1-\omega_1  \right )}{r_1 \Gamma(a_1)}\int_{0}^{\infty}\frac{1}{(z+\gamma)}
\exp\left ( -\frac{1}{\lambda_2} \left ( \frac{C\gamma }{\mu_{r_{2}}z} \right )^{\frac{1}{r_2}}\right )
{\rm{G}}_{0,1}^{1,0}\left[ \frac{1}{b_1^{c_1}}\left (\frac{ z+\gamma}{\mu_{r_1}} \right )^{\frac{c_1}{r_1}}\left| \begin{matrix} {-} \\ {a_1} \\ \end{matrix} \right.\right]dz.
\end{align}

\begin{align}\label{I4P1}
 \mathcal{I}_4=\frac{c_1 \left (1-\omega_1  \right )\left (1-\omega_2  \right )}{r_1 \Gamma(a_1)\Gamma(a_2)}\int_{0}^{\infty}\frac{1}{(z+\gamma)}
{\rm{G}}_{0,1}^{1,0}\left[ \frac{1}{b_1^{c_1}}\left (\frac{ z+\gamma}{\mu_{r_1}} \right )^{\frac{c_1}{r_1}}\left| \begin{matrix} {-} \\ {a_1} \\ \end{matrix} \right.\right]
{\rm{G}}_{1,2}^{2,0}\left[ \frac{1}{b_2^{c_2}}\left (\frac{ C \gamma}{\mu_{r_2}z} \right )^{\frac{c_2}{r_2}}\left| \begin{matrix} {1} \\ {a_2,0} \\ \end{matrix} \right.\right]dz.
\end{align}

Applying \cite[Eq.(8.4.3/1)]{PrudinkovVol3} to transform the exponential
function into its correspondent Meijer's G function, utilizing \cite[Eq.(9.31/2)]{Tableofintegrals} to inverse the argument of the Meijer's G function, then using the primary definition of the Meijer's G function in \cite[Eq.(9.301)]{Tableofintegrals}, $\mathcal{I}_1$ can be expressed as

\begin{align}\label{I1P2}
\nonumber \mathcal{I}_1&=\frac{\omega_1 \omega_2}{r_1}\frac{1}{(2\pi i)^{2}}\int\limits_{\mathcal{C}_1}\int\limits_{\mathcal{C}_2}
\Gamma(1-s)\Gamma(t) \left ( \frac{1}{\lambda_1 \mu_{r_{1}}^{\frac{1}{r_1}}}\right ) ^{s}\\
&\times\left ( \lambda_2 \left ( \frac{\mu_{r_2}}{C \gamma} \right ) ^{\frac{1}{r_2}}\right )^t
 \int_{0}^{\infty}(z+\gamma)^{\frac{s}{r_1}-1}\,z^{\frac{t}{r_2}}\,dz\,ds\,dt,
\end{align}
where $\mathcal{C}_1$ and $\mathcal{C}_2$ represent the $s$-plane and the $t$-plane contours, respectively.

Applying the integral identity \cite[Eq.(3.194/3)]{Tableofintegrals}, using \cite[Eq.(8.384/1)]{Tableofintegrals}, then making the change of variables $s=s/r_1$, $t=t/r_2$ with some algebraic manipulations, lead us to the following expression of $\mathcal{I}_1$

\begin{align}\label{I1P3}
\mathcal{I}_1= \frac{\omega_1 \omega_2}{(2\pi i)^{2}}\int\limits_{\mathcal{C}_1}\int\limits_{\mathcal{C}_2}
\frac{\Gamma(1+r_1s)}{\Gamma(1+s)}\Gamma(-t)\Gamma(1-r_2t)\Gamma(s+t)
\left ( \frac{\lambda_1^{r_1} \mu_{r_1}}{\gamma}\right ) ^{s}\left (\frac{C }{\lambda_2^{r_2} \mu_{r_2 }}  \right )^{t}ds\,dt,
\end{align}
which can be represented in terms of the extended generalized bivariate Fox's H function, ${\rm{H}_{\cdot,\cdot:\cdot,\cdot:\cdot,\cdot}^{\cdot,\cdot:\cdot,\cdot:\cdot,\cdot}}$ by means of using \cite[Eq.(1.1)]{HFoxIntegrals} as
\begin{align}\label{I1P4}
\mathcal{I}_1= \omega_1 \omega_2{\rm{H}}_{1,0:1,1:0,2}^{0,1:0,1:2,0}\begin{bmatrix}
\begin{matrix}
(1;1,1)\\--
\end{matrix}
\Bigg|\begin{matrix}
(0,r_1)\\(0,1)
\end{matrix}
\Bigg|\begin{matrix}
--\\(0,1)(1,r_2)
\end{matrix}
\Bigg|
\frac{\lambda_1^{r_1} \mu_{r_1}}{\gamma},\frac{C}{\lambda_2^{r_2} \mu_{r_2}}
\end{bmatrix}.
\end{align}

Note that when the first UWOC link is operating under the heterodyne detection technique, i.e. $r_1=1$, (\ref{I1P3}) simplifies to
\begin{align}\label{I1R1P1}
\mathcal{I}_1= \omega_1 \omega_2
\frac{1}{2\pi i}\int\limits_{\mathcal{C}_2}\Gamma(-t)\Gamma(1-r_2t)
\left (\frac{C }{\lambda_2^{r_2} \mu_{r_2 }}  \right )^{t}
 \frac{1}{2\pi i}
\int\limits_{\mathcal{C}_1}
\Gamma(s+t) \left ( \frac{\lambda_1 \mu_{r_1}}{\gamma}\right ) ^{s}ds\,dt.
\end{align}
By using the definition of the Meijer's G function in \cite[Eq.(9.301)]{Tableofintegrals}, (\ref{I1R1P1}) can be written as

\begin{align}\label{I1R1P2}
\mathcal{I}_1=\frac{ \omega_1 \omega_2}{2\pi i}
\int\limits_{\mathcal{C}_2}\Gamma(-t)\Gamma(1-r_2t)
{\rm{G}}_{1,0}^{0,1}\left[ \frac{\lambda_1 \mu_{r_1}}{\gamma}\left| \begin{matrix} {1-t} \\ {-} \\ \end{matrix} \right.\right]\left (\frac{C }{\lambda_2^{r_2} \mu_{r_2 }}  \right )^{t}dt.
\end{align}
Now, applying \cite[Eq.(07.34.03.0046.01)]{Wolfram}, we get

\begin{align}\label{I1R1P3}
\mathcal{I}_1= \omega_1 \omega_2 e^{ -\frac{\gamma}{\lambda_1 \mu_{r_{1}}} }
\frac{1}{2\pi i}\int\limits_{\mathcal{C}_2}\Gamma(-t)\Gamma(1-r_2t)
\left (\frac{C \gamma}{\lambda_1\lambda_2^{r_2} \mu_{r_{1} }\mu_{r_{2 }}}  \right )^{t}dt,
\end{align}\noindent
which can be easily expressed in terms of the Fox's H function in the special case of $r_1=1$ by utilizing \cite[Eq.(1.2)]{HFunction} as
\begin{align}\label{I1R1P4}
\mathcal{I}_1= \omega_1 \omega_2 e^{ -\frac{\gamma}{\lambda_1 \mu_{r_{1}}} }
{\rm{H}}_{0,2}^{2,0}\left[\frac{C \gamma}{\lambda_1\lambda_2^{r_2} \mu_{r_{1} }\mu_{r_{2 }}}\left| \begin{matrix} {--} \\ {(0,1)(1,r_{2})} \\ \end{matrix} \right. \right].
\end{align}

By utilizing \cite[Eq.(9.31/2) and (9.301)]{Tableofintegrals} then \cite[Eqs.(3.194/3) and (8.384/1)]{Tableofintegrals} followed by the change of variables $s=s/r_1$, $t=c_2 t/r_2$, and finally applying \cite[Eq. (1.1)]{HFoxIntegrals}, we get the following closed-form expression of $\mathcal{I}_2$ in terms of the bivariate Fox's H function as

\begin{align}\label{I2P2}
\mathcal{I}_2= \frac{\omega_1 \left (1-\omega_2  \right )}{\Gamma(a_2)}
{\rm{H}}_{1,0:1,1:0,2}^{0,1:0,1:2,0}\begin{bmatrix}
\begin{matrix}
(1;1,1)\\--
\end{matrix}
\Bigg|\begin{matrix}
(0,r_1)\\(0,1)
\end{matrix}
\Bigg|\begin{matrix}
--\\(0,1)(a_2,\frac{r_2}{c_2})
\end{matrix}
\Bigg|
\frac{\lambda_1^{r_1} \mu_{r_1}}{\gamma},\frac{C}{b_2^{r_2} \mu_{r_2}}
\end{bmatrix}.
\end{align}

In the special case where the first hop undergoes the heterodyne detection technique (i.e. $r_1=1$) and similar to $\mathcal{I}_1$, $\mathcal{I}_2$ reduces to

\begin{align}\label{I2R1}
\mathcal{I}_2= \frac{\omega_1 \left (1-\omega_2  \right ) }{\Gamma(a_2)}e^{ -\frac{\gamma}{\lambda_1 \mu_{r_{1}}} }
{\rm{H}}_{0,2}^{2,0}\left[\frac{C \gamma}{\lambda_1b_2^{r_2} \mu_{r_{1} }\mu_{r_{2 }}}\left| \begin{matrix} {--} \\ {(0,1)(a_2,\frac{r_{2}}{c_2})} \\ \end{matrix} \right. \right].
\end{align}

By using \cite[Eq.(9.31/2) and (9.301)]{Tableofintegrals} and \cite[Eqs.(3.194/3) and (8.384/1)]{Tableofintegrals} then the change of variables $s=c_1s/r_1$, $t=t/r_2$, and applying \cite[Eq. (1.1)]{HFoxIntegrals} with some simplifications, an closed-form expression for $\mathcal{I}_3$ can be obtained as

\begin{align}\label{I3P2}
\mathcal{I}_3= \frac{\omega_2 \left (1-\omega_1  \right )}{\Gamma(a_1)}
 {\rm{H}}_{1,0:1,1:0,2}^{0,1:0,1:2,0}\begin{bmatrix}
\begin{matrix}
(1;1,1)\\--
\end{matrix}
\Bigg|\begin{matrix}
\left (1-a_1,\frac{r_1}{c_1}  \right )\\(0,1)
\end{matrix}
\Bigg|\begin{matrix}
--\\(0,1)(1,r_2)
\end{matrix}
\Bigg|
\frac{b_1^{r_1} \mu_{r_1}}{\gamma},\frac{C}{\lambda_2^{r_2} \mu_{r_2}}
\end{bmatrix}.
\end{align}

Finally, utilizing \cite[Eq.(9.31/2), (9.301), (3.194/3), and (8.384/1)]{Tableofintegrals} along with the change of variables $s=c_1s/r_1$, $t=c_2 t/r_2$ then employing \cite[Eq. (1.1)]{HFoxIntegrals} results in the following expression of $\mathcal{I}_4$

\begin{align}\label{I4P2}
 \mathcal{I}_4= \frac{\left (1-\omega_1  \right )\left ( 1-\omega_2  \right )}{\Gamma(a_1)\Gamma(a_2)}
  {\rm{H}}_{1,0:1,1:0,2}^{0,1:0,1:2,0}\begin{bmatrix}
\begin{matrix}
(1;1,1)\\--
\end{matrix}
\Bigg|\begin{matrix}
\left (1-a_1,\frac{r_1}{c_1}  \right )\\(0,1)
\end{matrix}
\Bigg|\begin{matrix}
--\\(0,1)\left (a_2,\frac{r_2}{c_2}  \right )
\end{matrix}
\Bigg|
\frac{b_1^{r_1} \mu_{r_1}}{\gamma},\frac{C}{b_2^{r_2} \mu_{r_2}}
\end{bmatrix}.
\end{align}\normalsize

Now, substituting (\ref{I1P4}), (\ref{I2P2}), (\ref{I3P2}), and (\ref{I4P2}) in (\ref{CDFUWOCP3}) yields the desired CDF expression of $\gamma$ given in terms of the bivariate Fox's H function in (\ref{SNRCDF}).

\setcounter{equation}{0}
\section{Proof of Corollary 1.1}\label{B}
By using the definition of the Fox's H function in \cite[Eq.(1.1)]{HFunction}, $\mathcal{I}_1$ can be written as

\begin{align}\label{I1HighSNR1}
\mathcal{I}_1= \frac{\omega_1 \omega_2}{2\pi i}\int\limits_{\mathcal{C}_1}
\frac{\Gamma(1+r_1s)}{\Gamma(1+s)}{\rm{H}}_{1,2}^{2,1}\left[\frac{C }{\lambda_2^{r_2}\mu_{r_{2 }}}\left| \begin{matrix} {(1-s,1)} \\ {(0,1)(1,r_{2})} \\ \end{matrix} \right. \right]
 \left ( \frac{\lambda_1^{r_1} \mu_{r_1}}{\gamma}\right )^{s}ds
\end{align}
For high values of $\mu_{r_{2 }}$ the Fox's H functions in (\ref{I1HighSNR1}) can be approximated by means of using the Taylor expansion of the Fox's H function in \cite[Eq.(1.8.4)]{HTranforms} and keeping the first terms as

\begin{align}\label{I1HighSNR2}
{\rm{H}}_{1,2}^{2,1}\left[\frac{C }{\lambda_2^{r_2}\mu_{r_{2 }}}\left| \begin{matrix} {(1-s,1)} \\ {(0,1)(1,r_{2})} \\ \end{matrix} \right. \right]\underset{\mu_{r_{2}}\to \infty}{\mathop{\approx }} \Gamma(s).
\end{align}
Substituting (\ref{I1HighSNR2}) into (\ref{I1HighSNR1}) and applying \cite[Eq.(1.1)]{HFunction}, $\mathcal{I}_1$ can be rewritten as

\begin{align}\label{I1HighSNR3}
\mathcal{I}_1\underset{\mu_{r_{2}}\to \infty}{\mathop{\approx }} \omega_1 \omega_2{\rm{H}}_{2,1}^{0,2}\left[\frac{\lambda_1^{r_1}\mu_{r_{1 }}}{\gamma}\left| \begin{matrix} {(0,r_1)(1,1)} \\ {(0,1)} \\ \end{matrix} \right. \right],
\end{align}
which can be further simplified by using \cite[Eq.(1.5.9)]{HTranforms} yielding
\begin{align}\label{I1HighSNR4}
\mathcal{I}_1\underset{\mu_{r_{1}},\mu_{r_{2}}\to \infty}{\mathop{\approx }} \omega_1 \omega_2
\left ( 1-\left (\frac{\gamma}{\lambda_1^{r_1}\mu_{r_{1 }}}  \right )^{\frac{1}{r_1}}\right ).
\end{align}

Following the same approach as in the case of $\mathcal{I}_1$ and applying \cite[Eq.(1.1)]{HFunction} then \cite[Eqs.(1.8.4) and (1.5.9)]{HTranforms} with some algebraic manipulations, we get the asymptotic expressions of $\mathcal{I}_2$, $\mathcal{I}_3$, and $\mathcal{I}_4$ as

\begin{align}\label{I2HighSNR}
\mathcal{I}_2\underset{\mu_{r_{1}},\mu_{r_{2}}\to \infty}{\mathop{\approx }} \omega_1 (1-\omega_2)\left ( 1-\left (\frac{\gamma}{\lambda_1^{r_1}\mu_{r_{1 }}}  \right )^{\frac{1}{r_1}} \right )
-\frac{r_1  \omega_1 (1-\omega_2) }{\Gamma(a_2+1)}\Gamma\left (1-\frac{r_1 a_2 c_2 }{r_2} \right )
\left ( \frac{C \gamma}{\lambda_1^{r_1} b_2^{r_2}\mu_{r_{1 }}\mu_{r_{2 }}} \right )^{\frac{a_2 c_2}{r_2}}.
\end{align}
\begin{align}\label{I3HighSNR}
 \mathcal{I}_3\underset{\mu_{r_{1}},\mu_{r_{2}}\to \infty}{\mathop{\approx }} \omega_2(1-\omega_1)-\frac{\omega_2(1-\omega_1)}{\Gamma(a_1+1)}\left ( \frac{\gamma}{b_1^{r_1}\mu_{r_{1}}} \right )^{\frac{a_1 c_1}{r_1}}
-\frac{\omega_2(1-\omega_1)}{\Gamma(a_1)}\Gamma\left ( a_1-\frac{r_1}{c_1 r_2} \right )\left ( \frac{C \gamma}{b_1^{r_1}\lambda_2^{r_2}\mu_{r_{1}}\mu_{r_{2}}} \right )^{\frac{1}{r_2}}.
\end{align}

\begin{align}\label{I4HighSNR}
\nonumber \mathcal{I}_4\underset{\mu_{r_{1}},\mu_{r_{2}}\to \infty}{\mathop{\approx }} &(1-\omega_1)(1-\omega_2)-\frac{(1-\omega_1)(1-\omega_2)}{\Gamma(a_1+1)}\left ( \frac{\gamma}{b_1^{r_1}\mu_{r_{1}}} \right )^{\frac{a_1 c_1}{r_1}}\\
&-\frac{(1-\omega_1)(1-\omega_2)}{\Gamma(a_1)\Gamma(a_2+1)}\Gamma\left ( a_1-\frac{r_1 a_2 c_2}{c_1 r_2} \right )\left ( \frac{C \gamma}{b_1^{r_1}b_2^{r_2}\mu_{r_{1}}\mu_{r_{2}}} \right )^{\frac{a_2 c_2}{r_2}}.
\end{align}

Substituting (\ref{I1HighSNR4}), (\ref{I2HighSNR}), (\ref{I3HighSNR}), and (\ref{I4HighSNR}) into (\ref{CDFUWOCP3}) with some simplifications, we get an accurate simple closed-form expression for the CDF at high SNR as shown by (\ref{CDFHighSNR}).

\setcounter{equation}{0}
\section{Proof of Theorem 2}\label{C}
The PDF of the end-to-end SNR can be obtained by differentiating (\ref{SNRCDF}) with respect to $\gamma$ as

\begin{align}\label{PDFP1}
\nonumber & f_{\gamma}(\gamma)=\frac{\omega_1 \omega_2}{\gamma (2 \pi i)^2}\int\limits_{\mathcal{C}_1}\int\limits_{\mathcal{C}_2}
\frac{\Gamma(1+r_1s)}{\Gamma(s)}\Gamma(-t)\Gamma(1-r_2t)\Gamma(s+t)
\left ( \frac{\lambda_1^{r_1} \mu_{r_1}}{\gamma}\right ) ^{s}\left (\frac{C }{\lambda_2^{r_2} \mu_{r_2 }}  \right )^{t}ds\,dt\\
\nonumber &+\frac{\omega_1 \left (1-\omega_2  \right )}{\Gamma(a_2)\gamma (2 \pi i)^2}\int\limits_{\mathcal{C}_1}\int\limits_{\mathcal{C}_2}
\frac{\Gamma(1+r_1s)}{\Gamma(s)}\Gamma(-t)\Gamma(s+t)\Gamma\left(a_2-\frac{r_2}{c_2}t\right)\left ( \frac{\lambda_1^{r_1} \mu_{r_1}}{\gamma}\right ) ^{s}\left (\frac{C }{b_2^{r_2} \mu_{r_2 }}  \right )^{t}ds\,dt\\
\nonumber &
+\frac{\omega_2 (1-\omega_1)}{\Gamma(a_1)\gamma (2 \pi i)^2}\int\limits_{\mathcal{C}_1}\int\limits_{\mathcal{C}_2}\frac{\Gamma\left(a_1+\frac{r_1}{c_1}s\right)}{\Gamma(s)}\Gamma(-t)\Gamma(1-r_2t)\Gamma(s+t)
\left ( \frac{b_1^{r_1} \mu_{r_1}}{\gamma}\right ) ^{s}\left (\frac{C }{\lambda_2^{r_2} \mu_{r_2 }}  \right )^{t}ds\,dt\\
&+\frac{ (1-\omega_1)(1-\omega_2)}{\Gamma(a_1)\Gamma(a_2)\gamma (2 \pi i)^2}\int\limits_{\mathcal{C}_1}\int\limits_{\mathcal{C}_2}
\frac{\Gamma\left(a_1+\frac{r_1}{c_1}s\right)}{\Gamma(s)}\Gamma(-t)\Gamma\left(a_2-\frac{r_2}{c_2}t\right) \Gamma(s+t)
\left ( \frac{b_1^{r_1} \mu_{r_1}}{\gamma}\right )^{s}\left (\frac{C }{b_2^{r_2} \mu_{r_2 }}  \right )^{t}ds\,dt.
\end{align} \noindent
Therefore, applying \cite[Eq.(1.1)]{HFoxIntegrals} we get the PDF in exact closed-form as shown in (\ref{SNRPDF}).

\setcounter{equation}{0}
\section{Proof of Theorem 3}\label{D}
Using \cite[Eq.(2.3)]{HFoxIntegrals}, the PDF of $\gamma$ can be formulated in terms of integrals involving the product of Fox's H functions as

\begin{align}\label{PDFProductH}
\nonumber & f_{\gamma}(\gamma)=\frac{\omega_1 \omega_2}{\gamma}\int_{0}^{\infty}\frac{e^{-x}}{x}
{\rm{H}}_{1,1}^{0,1}\left[\frac{\lambda_1^{r_1}\mu_{r_{1 }}x}{\gamma}\left| \begin{matrix} {(0,r_1)} \\ {(1,1)} \\ \end{matrix} \right. \right]{\rm{H}}_{0,2}^{2,0}\left[\frac{C x }{\lambda_2^{r_2}\mu_{r_{2 }}}\left| \begin{matrix} {--} \\ {(0,1)(1,r_{2})} \\ \end{matrix} \right. \right]dx\\
\nonumber & +\frac{\omega_1 \left (1-\omega_2  \right )}{\Gamma(a_2)\gamma}
\int_{0}^{\infty}\frac{e^{-x}}{x}{\rm{H}}_{1,1}^{0,1}\left[\frac{\lambda_1^{r_1}\mu_{r_{1 }}x}{\gamma}\left| \begin{matrix} {(0,r_1)} \\ {(1,1)} \\ \end{matrix} \right. \right]{\rm{H}}_{0,2}^{2,0}\left[\frac{C x }{b_2^{r_2}\mu_{r_{2 }}}\left| \begin{matrix} {--} \\ {(0,1)(a_2,\frac{r_{2}}{c_2})} \\ \end{matrix} \right. \right]dx\\
\nonumber &+\frac{\omega_2 \left (1-\omega_1  \right )}{\Gamma(a_1)\gamma}
\int_{0}^{\infty}\frac{e^{-x}}{x}
{\rm{H}}_{1,1}^{0,1}\left[\frac{b_1^{r_1}\mu_{r_{1 }}x}{\gamma}\left| \begin{matrix} {(1-a_1,\frac{r_1}{c_1})} \\ {(1,1)} \\ \end{matrix} \right. \right]{\rm{H}}_{0,2}^{2,0}\left[\frac{C x }{\lambda_2^{r_2}\mu_{r_{2 }}}\left| \begin{matrix} {--} \\ {(0,1)(1,r_2)} \\ \end{matrix} \right. \right]dx\\
 &  +\frac{\left (1-\omega_1  \right )\left (1-\omega_2  \right )}{\Gamma(a_1)\Gamma(a_2)\gamma}
\int_{0}^{\infty}\frac{e^{-x}}{x}{\rm{H}}_{1,1}^{0,1}\left[\frac{b_1^{r_1}\mu_{r_{1 }}x}{\gamma}\left| \begin{matrix} {(1-a_1,\frac{r_1}{c_1})} \\ {(1,1)} \\ \end{matrix} \right. \right]{\rm{H}}_{0,2}^{2,0}\left[\frac{C x }{b_2^{r_2}\mu_{r_{2 }}}\left| \begin{matrix} {--} \\ {(0,1)(a_2,\frac{r_2}{c_2})} \\ \end{matrix} \right. \right]dx.
\end{align}

Substituting (\ref{PDFProductH}) into the definition of the moments and utilizing \cite[Eq.(1.58)]{HFunction} then applying \cite[Eq.(2.8)]{HFunction}, the moments can be expressed as

\begin{align}\label{MomentsP1}
\nonumber &\mathbb{E}[\gamma^n]=\left [\frac{\omega_1 \omega_2 }{\Gamma(n)}\Gamma(1+r_1n) \left ( \lambda_1^{r_1}\mu_{r_{1}} \right )^n+\frac{\omega_2 \left (1-\omega_1  \right ) }{\Gamma(a_1)\Gamma(n)} \Gamma\left (a_1+\frac{r_1}{c_1}n  \right )
\left ( b_1^{r_1}\mu_{r_{1}} \right )^n \right ]
\int_{0}^{\infty}\frac{e^{-x}}{x^{1-n}}\\
\nonumber & \times{\rm{H}}_{0,2}^{2,0}\left[\frac{C x }{\lambda_2^{r_2}\mu_{r_{2 }}}\left| \begin{matrix} {--} \\ {(0,1)(1,r_{2})} \\ \end{matrix} \right. \right]dx+\int_{0}^{\infty}\frac{e^{-x}}{x^{1-n}}{\rm{H}}_{0,2}^{2,0}\left[\frac{C x }{b_2^{r_2}\mu_{r_{2 }}}\left| \begin{matrix} {--} \\ {(0,1)\left (a_2,\frac{r_2}{c_2}  \right )} \\ \end{matrix} \right. \right]dx\\
 & \times \left [\frac{\omega_1 \left (1-\omega_2  \right )}{\Gamma(a_2)\Gamma(n)}\Gamma(1+r_1n)   \left ( \lambda_1^{r_1}\mu_{r_{1}} \right )^n+\frac{ \left (1-\omega_1  \right )\left (1-\omega_2  \right ) }{\Gamma(a_1)\Gamma(a_2)\Gamma(n)} \Gamma\left (a_1+\frac{r_1}{c_1}n  \right ) \left ( b_1^{r_1}\mu_{r_{1}} \right )^n \right].
\end{align}
Finally, by means of employing \cite[Eq.(2.8.4)]{HTranforms} the moments can be obtained in terms of the Fox's H function as shown in (\ref{moments}).



\setcounter{equation}{0}
\section{Proof of Theorem 4}\label{E}
Substituting (\ref{PDFP1}) in (\ref{BERdef}) then integrating
using \cite[Eq.(3.381/4)]{Tableofintegrals}, the BER can be written as

\begin{align}
\nonumber &P_e=\delta\sum_{k=1}^{n}\left\{\frac{1}{2}-\frac{\omega_1 \omega_2}{2\Gamma(p) \left (  2 \pi i\right )^2}\int\limits_{\mathcal{C}_1}\int\limits_{\mathcal{C}_2}
\frac{\Gamma(1+r_1s)\Gamma(p-s)}{\Gamma(1+s)}\Gamma(-t)\Gamma(1-r_2t)\Gamma(s+t)
\left (q_k \lambda_1^{r_1} \mu_{r_1}\right ) ^{s}\right.\\
\nonumber &\times\left. \left (\frac{C }{\lambda_2^{r_2} \mu_{r_2 }}  \right )^{t}ds\,dt -\frac{\omega_1 \left (1-\omega_2  \right )}{2\Gamma(a_2)\Gamma(p) \left ( 2\pi i \right )^2}
 \int\limits_{\mathcal{C}_1}\int\limits_{\mathcal{C}_2}\frac{\Gamma(1+r_1s)\Gamma(p-s)}{\Gamma(1+s)}\Gamma(-t)\Gamma\left(a_2-\frac{r_2}{c_2}t\right)\Gamma(s+t)
\right.\\
\nonumber & \times \left.\left (q_k \lambda_1^{r_1}\mu_{r_1}\right ) ^{s} \left (\frac{C }{b_2^{r_2} \mu_{r_2 }}  \right )^{t}ds\,dt
-\frac{\omega_2 \left (1-\omega_1  \right )}{2\Gamma(a_1)\Gamma(p) \left (  2 \pi i\right )^2}\int\limits_{\mathcal{C}_1}\int\limits_{\mathcal{C}_2}
\frac{\Gamma\left(a_1+\frac{r_1}{c_1}s\right)\Gamma(p-s)}{\Gamma(1+s)}\Gamma(-t)\Gamma(1-r_2t)\right.\\
\nonumber & \times
\left.\Gamma(s+t)\left (q_k b_1^{r_1} \mu_{r_1}\right ) ^{s}\left (\frac{C }{\lambda_2^{r_2} \mu_{r_2 }}  \right )^{t}ds\,dt
-\frac{(1-\omega_1) \left (1-\omega_2  \right )}{2\Gamma(a_1)\Gamma(a_2)\Gamma(p) \left ( 2\pi i \right )^2}\int\limits_{\mathcal{C}_1}\int\limits_{\mathcal{C}_2}\frac{\Gamma(a_1+\frac{r_1}{c_1}s)\Gamma(p-s)}{\Gamma(1+s)}\right.\\
& \times \left.\Gamma(-t)\Gamma\left(a_2-\frac{r_2}{c_2}t\right)\Gamma(s+t)\left (q_k b_1^{r_1}\mu_{r_1}\right ) ^{s}\left (\frac{C }{b_2^{r_2} \mu_{r_2 }}  \right )^{t}ds\,dt\right\}.
\end{align} \noindent
Applying \cite[Eq.(1.1)]{HFoxIntegrals}, the average BER can be derived in terms of the bivariate Fox's H function as shown in (\ref{SNRBER}).


\setcounter{equation}{0}
\section{Proof of Theorem 5}\label{F}
In order to obtain the ergodic capacity of $\gamma$, we first substitute (\ref{PDFP1}) in (\ref{CAPDEF}) then utilizing the Meijer's G representation of $\ln(1 + \tau \gamma)$ as ${\rm{G}}_{2,2}^{1,2}\left[\tau \gamma\left| \begin{matrix} {1,1} \\ {1,0} \\ \end{matrix} \right. \right]$ \cite[Eq.(07.34.03.0456.01)]{Wolfram}, and applying the integral identity \cite[Eq.(7.811/4)]{Tableofintegrals} yielding

\begin{align}\label{CapacityP1}
\nonumber &\overline{C}=\frac{\omega_1 \omega_2}{ (2 \pi i)^2}\int\limits_{\mathcal{C}_1}\int\limits_{\mathcal{C}_2}
\frac{\Gamma(1+r_1s)\Gamma(s)\Gamma(1-s)}{\Gamma(1+s)}\Gamma(-t)\Gamma(1-r_2t)\Gamma(s+t)\left (\tau \lambda_1^{r_1} \mu_{r_1}\right ) ^{s} \left (\frac{C }{\lambda_2^{r_2} \mu_{r_2 }}  \right )^{t}ds\,dt
\\
\nonumber & +\frac{\omega_1 \left (1-\omega_2  \right )}{\Gamma(a_2) (2 \pi i)^2}\int\limits_{\mathcal{C}_1}\int\limits_{\mathcal{C}_2}
\frac{\Gamma(1+r_1s)\Gamma(s)}{\Gamma(1+s)}
\Gamma(1-s)\Gamma(-t)\Gamma\left(a_2-\frac{r_2}{c_2}t\right)\Gamma(s+t)
\left ( \tau\lambda_1^{r_1} \mu_{r_1}\right ) ^{s}\left (\frac{C }{b_2^{r_2} \mu_{r_2 }}  \right )^{t}\\
\nonumber & \times ds\,dt+\frac{\omega_2 (1-\omega_1)}{\Gamma(a_1) (2 \pi i)^2}
\int\limits_{\mathcal{C}_1}\int\limits_{\mathcal{C}_2}
\frac{\Gamma\left(a_1+\frac{r_1}{c_1}s\right)\Gamma(s)\Gamma(1-s)}{\Gamma(1+s)}\Gamma(-t)\Gamma(1-r_2t)\Gamma(s+t)\left (\tau b_1^{r_1} \mu_{r_1}\right ) ^{s}\\
\nonumber & \times
 \left (\frac{C }{\lambda_2^{r_2} \mu_{r_2 }}  \right )^{t}ds\,dt+\frac{(1-\omega_1)(1-\omega_2)}{\Gamma(a_1)\Gamma(a_2) (2 \pi i)^2}
\int\limits_{\mathcal{C}_1}\int\limits_{\mathcal{C}_2}
\frac{\Gamma\left(a_1+\frac{r_1}{c_1}s\right)\Gamma(s)}{\Gamma(1+s)}
 \Gamma(1-s)\Gamma(-t)\Gamma\left(a_2-\frac{r_2}{c_2}t\right)\\
 & \times \Gamma(s+t) \left (\tau b_1^{r_1} \mu_{r_1}\right ) ^{s}\left (\frac{C }{b_2^{r_2} \mu_{r_2 }}  \right )^{t}ds\,dt.
\end{align} \noindent
Then exploiting (1.1) from \cite{HFoxIntegrals}, the ergodic capacity of dual-hop UWOC systems can be obtained in closed-form in terms of the bivariate Fox's H function in (\ref{SNRCapacity}).

\bibliographystyle{IEEEtran}
\bibliography{IEEEabrv,IEEEexample}
\end{document}